\documentclass[11pt]{amsart}
\usepackage[letterpaper,margin=1in]{geometry}                
\usepackage{amssymb,amsmath,amsfonts,stmaryrd}
\usepackage{xcolor}
\usepackage[colorlinks,
             linkcolor=blue,
             citecolor=green!70!black,
             pdfproducer={pdfLaTeX},
             pdfpagemode=None,
             bookmarksopen=true
             bookmarksnumbered=true]{hyperref}
\usepackage{graphicx}
\usepackage{multicol}

\usepackage{tikz}
\usetikzlibrary{decorations.markings,arrows,decorations.pathreplacing} 
\usepackage{arydshln}

\allowdisplaybreaks

\newcommand\arXiv[1]{\href{http://arxiv.org/abs/#1}{\nolinkurl{arXiv:#1}}}
\newcommand\DOI[1]{\href{http://dx.doi.org/#1}{\nolinkurl{DOI:#1}}}
\newcommand\MAILTO[1]{\href{mailto:#1}{\nolinkurl{#1}}}

\newtheorem{dummy}{Dummy}[section]

\newtheorem{theorem}[dummy]{Theorem}

\newtheorem{proposition}[dummy]{Proposition}
\newtheorem{conjecture}[dummy]{Conjecture}

\newtheorem{question}[dummy]{Question}
\newtheorem{definition}[dummy]{Definition}

\newtheorem*{mainquestion}{Main Question}

\theoremstyle{definition}

\newtheorem{examplenodiamond}[dummy]{Example}
\newtheorem{remarknodiamond}[dummy]{Remark}

\newenvironment{example}{\begin{examplenodiamond}}{\hfill\ensuremath{\Diamond}\end{examplenodiamond}}
\newenvironment{remark}{\begin{remarknodiamond}}{\hfill\ensuremath{\Diamond}\end{remarknodiamond}}

\usepackage{dsfont}
\renewcommand\mathbb\mathds

\newcommand\bC{\mathbb C}

\newcommand\bF{\mathbb F}

\newcommand\bM{\mathbb M}
\newcommand\bN{\mathbb N}

\newcommand\bZ{\mathbb Z}

\newcommand\cC{\mathcal C}

\newcommand\cN{\mathcal N}

\newcommand\rB{\mathrm B}

\newcommand\rH{\mathrm H}

\newcommand\rU{\mathrm U}

\usepackage[mathscr]{euscript}

\DeclareMathOperator\homology{H}
\renewcommand\H{\homology}

\newcommand\longto\longrightarrow
\newcommand\mono\hookrightarrow
\newcommand\epi\twoheadrightarrow

\newcommand\<\langle
\renewcommand\>\rangle
\newcommand\sminus\smallsetminus

\newcommand\st{\text{ s.t.\ }}

\newcommand\Leech{\mathrm{Leech}}

\newcommand\SU{\mathrm{SU}}

\newcommand\Aut{\mathrm{Aut}}

\DeclareMathOperator\Fer{Fer}
\DeclareMathOperator\Bos{Bos}

\DeclareMathOperator\Sq{Sq}

\newcommand\tr{\mathrm{tr}}

\DeclareMathOperator\End{End}

\newcommand\TMF{\mathrm{TMF}}
\newcommand\MF{\mathrm{MF}}

\newcommand\ev{\mathrm{ev}}
\newcommand\odd{\mathrm{odd}}

\newcommand\gp{\mathrm{gp}}
\DeclareMathOperator\SH{SH}
\DeclareMathOperator\rSH{rSH}

\DeclareMathOperator\Index{Index}

\DeclareMathOperator\CWE{CWE}

\newcommand\OL{\mathrm{OddLeech}}

\newcommand\Niem{{\operatorname{Nie}(A_1^{24})}}

\newcommand\define[1]{\emph{#1}}
\newcommand\cat[1]{\textsc{#1}}

\newcommand\mqlink{\hyperref[mainquestion]{Main Question}}

\title{Holomorphic SCFTs with small index}
\author{Davide Gaiotto and Theo Johnson-Freyd}
\address{Perimeter Institute for Theoretical Physics, Waterloo, Ontario}
\thanks{We thank Noam D.\ Elkies for many valuable conversations, some of which were hosted by \url{mathoverflow.net}, and for helping with a number of calculations. We also thank Greg Moore and Jeff Harvey for comments on a draft of this paper.
Research at the Perimeter Institute for Theoretical Physics is supported by the Government of Canada through the Department of Innovation, Science and Economic Development Canada and by the Province of Ontario through the Ministry of Research, Innovation and Science.}

\begin{document}

\begin{abstract}
We observe that every self-dual ternary code determines a holomorphic $\cN=1$ superconformal field theory. This provides ternary constructions of some well-known holomorphic $\cN=1$ SCFTs, including Duncan's ``supermoonshine'' model and the fermionic ``beauty and the beast'' model of Dixon, Ginsparg, and Harvey.
Along the way, we clarify some issues related to orbifolds of fermionic holomorphic CFTs.
We give a simple coding-theoretic description of the supersymmetric index and conjecture that for every self-dual ternary code this index is divisible by $24$; we are able to prove this conjecture except in the case when the code has length $12$ mod $24$. Lastly, we discuss a conjecture of Stolz and Teichner relating $\cN=1$ SCFTs with Topological Modular Forms. This conjecture implies constraints on the supersymmetric indexes of arbitrary holomorphic SCFTs, and suggests (but does not require) that there should be, for each~$k$, a holomorphic $\cN=1$ SCFT of central charge $12k$ and index $24/\gcd(k,24)$. We give ternary code constructions of SCFTs realizing this suggestion for $k\leq 5$.
\end{abstract}
\maketitle

The motivation for this note comes from an attempt to construct two-dimensional superconformal field theories (SCFTs) that would explain some features of the generalized cohomology theory known as Topological Modular Forms (TMF). The connection between SCFTs and TMF is predicted by conjectures of Stolz and Teichner \cite{MR2742432} proposing a geometric model of TMF in terms of (not necessarily conformal) two-dimensional supersymmetric field theory. In particular, as we explain in Section~\ref{TMF section}, each holomorphic vertex operator algebra (VOA) of central charge $c$ equipped with $\cN=1$ supersymmetry should determine a class in $\TMF_{2c} = \pi_{2c}\TMF(\mathrm{pt})$. Let $\MF_c$ denote the space of integral modular forms of weight $c$. There is an edge map $\TMF_{2c} \to \MF_c$, which in the proposed model should take an SCFT to its supersymmetric index, multiplied by $\eta^{2c} = \Delta^{c/12}$. A curious feature of $\TMF$ is that the smallest multiple of $\Delta^k$ in the image of $\TMF_{24k} \to \MF_{12k}$ is~$\frac{24}{\gcd(k,24)}\Delta^k$.
 Our goal is to construct holomorphic $\cN=1$ VOAs realizing these values. We pose this goal as our Main Question:
\begin{mainquestion}\label{mainquestion}
For each $k$, does there exist a holomorphic $\cN=1$ VOA with central charge $c=12k$ and supersymmetric index $\frac{24}{\gcd(k,24)}$?
\end{mainquestion}

In order to answer the \mqlink, and to explore TMF more generally, it is useful to have a source for holomorphic $\cN=1$ VOAs.  As described in Section~\ref{SCFT section}, every linear ternary self-dual, also called ``Type III,'' code determines a VOA equipped with $\cN=1$ structure; 
it is a lattice VOA for what in \cite{MR2522420} is called a ``$3$-framed lattice,'' and so we will refer to the result as a ``$3$-framed VOA.''
 Further $\cN=1$ VOAs can be constructed by orbifolding $3$-framed VOAs by nonanomalous symmetries.
We study the supersymmetry-preserving automorphisms of $3$-framed VOAs in Section~\ref{symmetry section}.
The theory of fermionic orbifolds is studied in Section~\ref{orbifold section}, where we clarify the role of the 't Hooft anomaly and of its trivializations.
 Among the $\cN=1$ VOAs that arise as orbifolds of $3$-framed VOAs are the ``supermoonshine'' model of \cite{MR2352133} (Examples~\ref{eg.scfts with c=12} and~\ref{supermoonshine continued}) and of the ``beauty and the beast'' model of \cite{MR968697} (Example~\ref{eg.bandb}).

The supersymmetric index of the $3$-framed VOA and of its orbifolds can be computed directly from the code (Proposition~\ref{prop.index=index} and Theorem~\ref{thm.answers}). This suggests that one study the ``index'' of self-dual ternary codes for their own sake. Our results and conjectures lead to the following  divisibility statement (Theorem~\ref{thm.index of a code} and Conjecture~\ref{code conjecture}): the index of a self-dual ternary code is always divisible by $24$. 
Our proof requires  nontrivial divisibility results about modular forms.
Section~\ref{code section} contains some elementary background on ternary codes and explains the definition of the index of a code without reference to supersymmetric field theory.

All together, one way to affirmatively answer the \mqlink\
is to construct self-dual ternary codes with  index $24$ (and an appropriate orbifoldable symmetry). 
For small $c$, the complete classification of self-dual ternary codes of length $c$ is known \cite{MR2522420}, but for large $c$ the best approach available is an exponential-time computer search.
With help from Noam D.\ Elkies implementing such a search, we succeeded at finding $\cN=1$ VOAs of central charge $c = 12k$ and supersymmetric index $\frac{24}{\gcd(k,24)}$ for $k\leq 5$.

This work suggests several interesting future directions of inquiry:
\begin{question}
Is there a systematic construction of self-dual ternary codes of  index $24$ which would avoid the use of intensive computation? Such a construction could provide a systematic construction of $\cN=1$ VOAs realizing the modular forms $\frac{24}{\gcd(k,24)}\Delta^k$.
\end{question}
\begin{question}
 After $\frac{24}{\gcd(k,24)}\Delta^k$, the next most interesting  modular forms in the image of $\TMF \to \MF$ are the Theta functions of even unimodular lattices. These will not come from holomorphic VOAs, but might come from full CFTs with $\cN=(0,1)$ supersymmetry. Is there a systematic construction that inputs an even unimodular lattice $\Lambda$ of rank $r$ and produces a $\cN=(0,1)$ full CFT of total central charge $c_R - c_L = r$ and Ramond--Ramond partition function $\Theta_\Lambda / \eta^r$?
\end{question}
\begin{question}
Does the Stolz--Teichner proposal generalize to involve supersymmetric field theories 
equipped with discrete symmetry group actions of specified 't Hooft anomaly? 
 It would be nice to find 
what the corresponding ``equivariant'' TMF theory 
should look like. The holomorphic SCFTs we build have large discrete symmetry groups and should 
have interesting images under such a conjectural correspondence.
\end{question}
\begin{question}
What is the physical meaning of TMF classes, especially of torsion type?
Which physical relation between two theories implies that they represent the same 
class in TMF?
\end{question}

\section{The index of a ternary code}\label{code section}

Let $\bF_3$ denote the field of order $3$.
By definition, a \define{ternary code} of length $c$ is a linear subspace  $C \subset \bF_3^c$. 
The ambient space $\bF_3^c$ has a standard ``Cartesian'' inner product 
$\langle (v_1,\dots,v_c), (w_1,\dots,w_c) \rangle = \sum_i v_i w_i \in \bF_3$,
 and so given a code $C \subset \bF_3^c$ one can construct its \define{dual code} $C^\perp = \{w \in \bF_3^c \st \langle w,v \rangle = 0 ~\forall v \in C\}$. A ternary code is \define{self-dual} if  $C^\perp = C$.

Elements of $\bF_3^c$ are called \define{words}, and if a code $C$ is fixed, elements of $C$ are called \define{code words}. The \define{Hamming weight} of a  word is its number of non-zero entries. A word is \define{maximal} if its Hamming weight is $c$.
In a self-dual ternary code, every code word is self-orthogonal, and so in particular has Hamming weight divisible by~$3$.
It is well-known that self-dual ternary codes of length $c$ can occur only when $c$ is divisible by~$4$, and so a self-dual ternary code  can contain maximal codewords only when $c = 12k$ for some $k\in \bN$.

Suppose $w \in \bF_3^c$ is a maximal word. Let us say that $w$ is \define{even} or \define{odd} according to the number of $1$s among its entries, mod $2$.

\begin{definition}\label{defn.indexofcode}
Suppose $C$ is a self-dual ternary code. The \define{index} of $C$ is
$$ \Index(C) = \# \{\text{even maximal codewords}\} - \#\{\text{odd maximal codewords}\}.$$
\end{definition}

\begin{remark}\label{rem.CWE}
The index is the value at $(x,y,z) = (0,1,-1)$ of the \define{complete weight enumerator} defined as
$$ \CWE_C(x,y,z) = \sum_{w\in C} x^{n_0(w)} y^{n_1(w)} z^{n_{-1}(w)} = \sum_{w \in C} \prod_{i=1}^c x_{w(i)}$$
where $n_i(w)$ is the number of $i$s among the coordinates of $w$, $(x_0,x_1,x_{-1}) = (x,y,z)$, and $w(i) \in \{0,1,-1\} = \bF_3$ is the $i$th coordinate of $w$.
\end{remark}

The first main result of this paper is:
\begin{theorem}\label{thm.index of a code}
  If $C$ is a self-dual ternary code of length divisible by $24$, then $\Index(C)$ is divisible by $24$. If the length of $C$ is merely divisible by $12$, then $\Index(C)$ is divisible by $12$.
\end{theorem}

If the length of $C$ is not divisible by $12$, then $\Index(C)$ vanishes, as there are simply no maximal codewords.
The second sentence in the Theorem follows immediately from the first together with the easy observation that $\Index(C' \oplus C'') = \Index(C') \Index(C'')$ (and the even easier observation that if $24$ divides $n^2$, then $12$ divides $n$). We will prove the first sentence in the next section, after Remark~\ref{remark.alternateproof}. Our proof relies on some nontrivial facts about Theta functions of lattices, and we do not know if there is an elementary proof. Furthermore, we expect that:

\begin{conjecture} \label{code conjecture}
  The index of any self-dual ternary code is divisible by $24$, even if the length is merely divisibly by $12$. For all lengths $c = 12k$, there exists a self-dual ternary code of index $24$.
\end{conjecture}

Conjecture~\ref{code conjecture} follows from the deeper Conjecture~\ref{holomorphic conjecture}, which, together with Proposition~\ref{prop.index=index}, implies that $\Index(C)$ is divisibly by $8$ whenever $C$ is a self-dual ternary code of length $c = 12 \bmod{24}$.

\begin{example}\label{eg.codes of length 12}
Up to signed coordinate permutations, there are only three self-dual ternary codes of length $c=12$:

\begin{enumerate}
\item
First is the direct sum of three copies of the unique code of length $4$. It can be generated by the matrix:
$$ \left( \begin{array}{cccccccccccc}
1 & 1 & -1 & 0  &  0 & 0 & 0 & 0  &  0 & 0 & 0 & 0 \\
0 & 1 & 1 & -1  &  0 & 0 & 0 & 0  &  0 & 0 & 0 & 0 \\
0 & 0 & 0 & 0  &  1 & 1 & -1 & 0  &  0 & 0 & 0 & 0 \\
0 & 0 & 0 & 0  &  0 & 1 & 1 & -1  &  0 & 0 & 0 & 0 \\
0 & 0 & 0 & 0  &  0 & 0 & 0 & 0  &  1 & 1 & -1 & 0 \\
0 & 0 & 0 & 0  &  0 & 0 & 0 & 0  &  0 & 1 & 1 & -1
\end{array} \right) $$
This code has no maximal codewords, hence index $0$. 

\item
Second, there is a self-dual ternary code of length $c=12$ with eight even and eight odd maximal codewords, hence index $0$. A generating matrix is:
$$ \left( \begin{array}{cccccccccccc}
1 & 1 & 1  &  0 & 0 & 0  &  0 & 0 & 0  &  0 & 0 & 0  \\
0 & 0 & 0  &  1 & 1 & 1  &  0 & 0 & 0  &  0 & 0 & 0  \\
0 & 0 & 0  &  0 & 0 & 0  &  1 & 1 & 1  &  0 & 0 & 0  \\
0 & 0 & 0  &  0 & 0 & 0  &  0 & 0 & 0  &  1 & 1 & 1  \\
0 & 1 & -1 &  0 & 1 & -1 &  0 & -1 & 1 &  0 & 0 & 0  \\
0 & 0 & 0  &  0 & 1 & -1 &  0 & 1 & -1 &  0 & -1 & 1
\end{array} \right) $$

\item
Finally, there is the Ternary Golay code from \cite{Golay}, which has 24 even and no odd maximal codewords, hence index $24$. It can be generated by the matrix:
$$ \left( \begin{array}{cccccccccccc}
1 & 0 & 0 & 0 & 0 & 0 & 0 & 1 & 1 & 1 & 1 & 1\\
0 & 1 & 0 & 0 & 0 & 0 & -1 & 0 & 1 & -1 & -1 & 1\\
0 & 0 & 1 & 0 & 0 & 0 & -1 & 1 & 0 & 1 & -1 & -1\\
0 & 0 & 0 & 1 & 0 & 0 & -1 & -1 & 1 & 0 & 1 & -1\\
0 & 0 & 0 & 0 & 1 & 0 & -1 & -1 & -1 & 1 & 0 & 1\\
0 & 0 & 0 & 0 & 0 & 1 & -1 & 1 & -1 & -1 & 1 & 0
\end{array} \right) $$
\end{enumerate}
In particular, Conjecture~\ref{code conjecture} holds for $c=12$.
\end{example}

\begin{example}
For $c=24$, we may confirm  Conjecture~\ref{code conjecture} by appealing
to the classification of self-dual ternary codes from  \cite{MR2522420}. The first step of their classification is to convert each self-dual ternary code $C$ of length $c$ into an odd unimodular lattice $\Lambda(C)$ of rank $c$; we will review that construction in the next section. Odd unimodular lattices of rank $\leq 24$ were classified by Borcherds and listed in \cite[Chapter 17]{MR1662447}. As we will explain during the proof of Theorem~\ref{thm.index of a code}, when $c=24$ the index of $C$ is precisely the difference between the numbers of roots of the two ``even neighbors'' of $\Lambda(C)$. Inspecting the list, we find that of the $338$ inequivalent self-dual ternary codes of length~$24$, $30$ have index~$24$.
\end{example}

Conjecture~\ref{code conjecture} is further supported by the following experimental evidence. Together with Noam D.\ Elkies, we randomly generated self-dual codes and calculated their indexes.  (The runtime of this calculation grows exponentially with the length $c$.) A random sampling of hundreds of self-dual ternary codes of length $c=36$ and of dozens of self-dual ternary codes of length $c=48$ and $c=60$ always produced codes of index divisible by $24$.  Codes of index precisely $24$ appeared fairly often. A coarse estimate predicts that the expected value of $\Index(C)^2$ grows slowly with $c$, contributing to the difficulty of finding codes of index $24$ by random search, but the total number of codes grows very quickly with $c$, so the existence of a code of index $24$ remains likely.

\section{From ternary codes to SCFTs}\label{SCFT section}

Suppose $C \subset \bF_3^c$ is a ternary code of length $c$ such that $C \subset C^\perp$. One can construct an integral lattice $\Lambda(C)$ from $C$ as follows. Consider the lattice $\sqrt3\bZ$, i.e.\ the rank-1 lattice with basis vector of square-length $3$, and its dual lattice $(\sqrt3\bZ)^* = \frac1{\sqrt3}\bZ$. 
Identify $\bF_3$ with the quotient $\frac1{\sqrt3}\bZ/\sqrt3\bZ$. Then $\Lambda(C)$ is defined by the following pullback:
$$ \begin{tikzpicture}[anchor=base]
  \path (0,0) node (C) {$C$}
    (0,2) node (L) {$\Lambda(C)$}
    (3,0) node (F) {$\bF_3^c$}
    (3,2) node (Z) {$\frac1{\sqrt3}\bZ^c$}
    (.5,1.5) node {$\lrcorner$};
  \draw[right hook->] (C) -- (C -| F.west);
  \draw[->>] (L) -- (C);
  \draw[right hook->] (L) -- (L -| Z.west);
  \draw[->>] (Z) -- (F);
\end{tikzpicture}$$
This lattice is integral because $C \subset C^\perp$. It is unimodular exactly when $C$ is self-dual.
Note that one automatically has an injection $\sqrt3\bZ^c \mono \Lambda(C)$. Following \cite{MR2522420}, we will call a rank-$c$ lattice $\Lambda$ equipped with an injection $\sqrt3\bZ^c \mono \Lambda$ \define{3-framed}. Each 3-framed lattice arises from a unique ternary code $C = \Lambda / \sqrt3\bZ^c$.

There is a well-known construction that produces from each integral lattice $\Lambda$ of rank $c$ a vertex operator algebra $V_\Lambda$ of central charge $c$.  For the definition of ``vertex operator algebra'' and for details of the construction, we refer the reader to the standard textbooks \cite{MR1651389,MR2082709}. In outline, the construction of $V_\Lambda$ goes as follows. One begins with a ``free boson'' VOA $\Bos(\mathfrak{h})$ built functorially from the vector space $\mathfrak{h} = L \otimes_\bZ \bC$. The irreducible vertex modules for $\Bos(\mathfrak{h})$ are naturally indexed by the points in $\mathfrak{h}$: given $\lambda \in \mathfrak{h}$, there is a module $M_\lambda$ generated over $\Bos(\mathfrak{h})$ by an element $\Gamma_\lambda \in M_\lambda$; as a vector space, $M_\lambda \cong \Bos(\mathfrak{h}) \otimes \bC\Gamma_\lambda$, where $\bC\Gamma_\lambda$ is the one-dimensional space with basis $\Gamma_\lambda$. The element $\Gamma_\lambda$ has conformal dimension $\lambda^2/2$.
Finally, we set
$$ V_\Lambda = \bigoplus_{\lambda \in \Lambda} (-1)^{\lambda^2} M_\lambda,$$
where the notation ``$(-1)^{\lambda^2}$'' denotes that the summands corresponding to vectors of odd square-length are considered to be fermionic, and those of even square-length are bosonic. If $\Lambda$ is an integral lattice, then $V_\Lambda$ is naturally a VOA \cite[Theorem 8.10.2]{MR996026}.
It is holomorphic (i.e.\ it has only one irreducible module) exactly when $\Lambda$ is unimodular.

The above description of $V_\Lambda$ can be made quite explicit when $\Lambda = \Lambda(C)$ for a ternary code $C \subset C^\perp \subset \bF_3^c$. Consider first the case $c=1$ and $C=\{0\}$, so that $\Lambda(C) = \sqrt3\bZ$.  Then $V_{\sqrt3\bZ}$ is generated by a ``free boson at level 3'' $\alpha$ together with fermionic operators $\Gamma_{\pm\sqrt3}$. The operator product expansions are
\begin{align*}
  \alpha(z)\,\alpha(w) &\sim \frac3{(z-w)^2} \\
  \alpha(z)\,\Gamma_{\pm\sqrt3}(w) &\sim \frac{\pm 3 \Gamma_{\pm\sqrt3}(w)}{z-w} \\
  \Gamma_{\pm\sqrt3}(z)\,\Gamma_{\pm\sqrt3}(w) &\sim 0 \quad\quad \text{(same sign)}\\
  \Gamma_{+\sqrt3}(z)\,\Gamma_{-\sqrt3}(w) & \sim \frac1{(z-w)^3} + \frac{\alpha(w)}{(z-w)^2} + \frac{\frac12 {:} \alpha^2(w) {:} + \frac12\partial\alpha(w)}{z-w}
\end{align*}
The subalgebra $\Bos(\sqrt3\bZ \otimes \bC)$ is generated by $\alpha$ alone. For general $\lambda = \pm n\sqrt3 \in \sqrt3\bZ$ with $n\geq 0$, we have $\Gamma_{\pm n\sqrt3} = {{:} \Gamma_{\pm \sqrt3}^n {:}}$.

When $c$ is arbitrary and $C=\{0\}$, we set
$$ V_{\sqrt3\bZ^c} = V_{\sqrt3\bZ}^{\otimes c}.$$
The irreps of $V_{\sqrt3\bZ}^{\otimes c}$ are naturally indexed by $(\sqrt3\bZ^c)^*/\sqrt3\bZ^c = \bF_3^c$ \cite{MR1245855}, and for arbitrary $C \subset \bF_3^c$ we set
$$ V_{\Lambda(C)} = \bigoplus_{w \in C} \left(V_{\sqrt3\bZ^c}\text{-module indexed by }w\right).$$

The VOA  $V_{\sqrt3\bZ}$ is also called ``$U(1)$ at level $3$.'' It is well known to be an $\cN=2$ minimal model: the boson $\alpha$ generates the R-symmetry, and the polarized supersymmetry generators are $\Gamma_{\pm\sqrt3}$. It follows that $V_{\Lambda(C)} \supset V_{\sqrt3\bZ}^{\otimes c}$  has the structure of $\cN=2$ supersymmetric VOA. We will focus on the induced $\cN=1$ structure. The generator of $\cN=1$ supersymmetry on $V_{\sqrt3\bZ}$ is $G = \Gamma_{+\sqrt3} + \Gamma_{-\sqrt3}$. Letting $G_i = 1 \otimes \dots \otimes 1 \otimes G \otimes 1\otimes \dots \otimes 1 \in V_{\sqrt3\bZ}^{\otimes c}$, with the $G$ in the $i$th spot, the generator of $\cN=1$ supersymmetry on $V_{\Lambda(C)}$ is $\sum_{i=1}^c G_i$. Based on the lattice notion from \cite{MR2522420}, we will use the term \define{3-framed VOA} for a VOA $V$ of central charge $c\in \bZ$ equipped with an injection $V_{\sqrt3\bZ}^{\otimes c} \mono V$. A 3-framed VOA $V$ is necessarily of the form $V_{\Lambda(C)}$ for some ternary code $C \subset C^\perp$.

\begin{example}\label{eg.scfts with c=3}
  Consider the non-self-dual code $C_3$ of length $3$ spanned by the vector $(1,1,1) \in \bF_3^3$. The corresponding lattice is $\Lambda(C_3) \cong A_2 \times \bZ$, where $A_2$ denotes the root lattice of $\mathfrak{sl}(3)$. Indeed, $\Lambda(C_3) \subset \frac1{\sqrt3}\bZ^3$ has two vectors of length $1$, namely $\pm\frac1{\sqrt3}(1,1,1)$; orthogonal to these are six vectors of square-length $2$, namely the cyclic permutations of $\pm\frac1{\sqrt3}(1,1,-2)$. We therefore find an isomorphism $V_{\Lambda(C_3)} \cong V_{A_2} \times V_\bZ$. The famous ``{boson-fermion correspondence}'' identifies the lattice VOA $V_\bZ$ with the VOA $\Fer(2)$ of two  Majorana fermions, and so $V_{\Lambda(C_3)} \cong V_{A_2} \times \Fer(2)$.
  
  For any rank-$r$ integral lattice $\Lambda$, the VOA $V_\Lambda \otimes \Fer(r)$ can be equipped with an $\cN=1$ supersymmetry that simply exchanges the $r$ free fermions generating $\Fer(r)$ with the $r$ free bosons generating $\Bos(r) = \Bos(\Lambda \otimes \bC) \subset V_\Lambda$ \cite{HeluaniKac2007}. In the basis we are working in, the $\cN=1$ supersymmetry on $V_{\Lambda(C_3)}\cong V_{A_2} \times \Fer(2)$ is not this one. However, we claim that there is an automorphism of $V_{\Lambda(C_3)}$ intertwining the two $\cN=1$ structures. In general, to write a rational VOA $V$ of central charge $c\in \bZ$ as a lattice VOA, one must choose a subalgebra $\Bos(c) \subset V$, i.e.\ one must choose $c$ many commuting fields of conformal weight $1$; the lattice is then the lattice of weights of the action of $\Bos(c)$ on~$V$~\cite{MR1325775}. (The choice of $\Bos(c) \subset V$ is analogous to a choice of Cartan subalgebra of a semisimple Lie algebra.) Starting with $V_{\Lambda(C_3)} \cong V_{A_2} \otimes V_\bZ$, we choose to keep the free boson pointing in the $(1,1,1)$-direction, but to replace the other two bosons (spanning the $(A_2 \otimes \bC)$-plane inside $\Lambda(C_3) \otimes \bC$) with $X = \Gamma_{\frac1{\sqrt3}(1,1,-2)} + \Gamma_{\frac1{\sqrt3}(1,-2,1)} + \Gamma_{\frac1{\sqrt3}(-2,1,1)}$ and $Y = \Gamma_{\frac1{\sqrt3}(-1,-1,2)} + \Gamma_{\frac1{\sqrt3}(-1,2,-1)} + \Gamma_{\frac1{\sqrt3}(2,-1,-1)}$. One may check that $X$ and $Y$ commute, and that the supersymmetry $G = \Gamma_{(\sqrt3,0,0)} + \dots + \Gamma_{(0,0,-\sqrt3)}$ exchanges $X$ and $Y$ with the free fermion fields $\Gamma_{\frac1{\sqrt3}(1,1,1)}$ and $\Gamma_{-\frac1{\sqrt3}(1,1,1)}$ respectively. Thus for this new Cartan we find the $\cN=1$ structure on $V_{A_2} \otimes \Fer(2)$ of  \cite{HeluaniKac2007}.
\end{example}

\begin{example}\label{eg.scfts with c=4}
  Let $C_4\subset \bF_3^4$ denote the unique self-dual ternary code of length~$4$, which has generator matrix
  $$ \left( \begin{array}{cccc} 1 & 1 & -1 & 0 \\ 0 & 1 & 1 & -1 \end{array} \right).$$
  Then $\Lambda(C_4) \cong \bZ^4$, and so $V_{\Lambda(C_4)} \cong \Fer(8)$ by the boson-fermion correspondence. $\cN=1$ supersymmetry structures on $\Fer(n)$ were classified in \cite{MR791865}, and correspond to $n$-dimensional semisimple Lie algebras. 
  There is a unique semisimple Lie algebra of dimension eight, namely $\mathfrak{su}(3)$, and so we must have an isomorphism of $\cN=1$ VOAs $V_{\Lambda(C_4)} \cong \Fer(\mathfrak{su}(3))$.
\end{example}

\begin{example}\label{eg.scfts with c=12}
  Let us work out the VOAs corresponding to the three self-dual ternary codes of length~$12$ from Example~\ref{eg.codes of length 12}.
  \begin{enumerate}
    \item The first code from Example~\ref{eg.codes of length 12} was $C_4^{\oplus 3}$, where $C_4$ is the unique self-dual ternary code of length~$4$ from Example~\ref{eg.scfts with c=4}. The $\cN=1$ structure on $V_{\Lambda(C_4 \oplus C_4 \oplus C_4)} \cong V_{\Lambda(C_4)}^{\otimes 3} \cong \Fer(8)^{\otimes 3} \cong \Fer(24)$ is the one coming from the Lie algebra $\mathfrak{su}(3)^3$.
    \item Let $C_3$ denote the non-self-dual code of length $3$ from Example~\ref{eg.scfts with c=3}, and let $C$ denote the second self-dual ternary code $C$ in Example~\ref{eg.codes of length 12}. Then $C$ is an extension of $C_3^{\oplus 4}$, and so $V_{\Lambda(C)}$ is an extension of $V_{\Lambda(C_3)}^{\otimes 4} = V_{A_2^4} \otimes \Fer(8)$.  The lattice $A_2^4$ has only one unimodular extension, namely the $E_8$ lattice.  After changing bases as in Example~\ref{eg.scfts with c=3}, we find an isomorphism of $\cN=1$ SCFTs between $V_{\Lambda(C)}$ and $V_{E_8} \otimes \Fer(8)$ made into an $\cN=1$ SCFT as in \cite{HeluaniKac2007}.
    \item Finally, consider the Ternary Golay code, listed third in Example~\ref{eg.codes of length 12}. The lattice $\Lambda(\text{Golay})$ turns out to be isomorphic to the $D_{12}^+$ lattice, i.e.\ the $D_{12}$ root lattice together with its coset containing the highest weight of the (positive) half-spin representation of $\operatorname{Spin}(24)$. (Indeed, $\Lambda(\text{Golay})$ contains no vectors of length~$1$, and $D_{12}^+$ is the unique unimodular lattice of rank~$12$ with this property.) 
    The corresponding VOA $V_{D_{12}^+}$ is isomorphic to the 
    Duncan's ``supermoonshine'' VOA $V^{f\natural}$ from \cite{MR2352133}. One of the main results of that paper is that $V^{f\natural}$ carries a unique-up-to-isomorphism $\cN=1$ supersymmetry structure (and is the unique $c=12$ and $\cN=1$ SCFT with no fields of conformal dimension $1/2$). See also Example~\ref{supermoonshine continued}.
  \end{enumerate}
\end{example}

Suppose $V$ is a holomorphic $\cN=1$ VOA of central charge $c$, and let $V_R$ denote its Ramond sector, i.e.\ the Hilbert space assigned by $V$ to the circle with non-bounding spin structure.
Consider the Ramond--Ramond partition function $Z_{RR}(V) = \tr_{V_R} (-1)^f q^{L_0 - c/24}$, i.e.\ the partition function of $V$ evaluated on elliptic curves with spin structure induced from the Lie group framing of the elliptic curve.
If $V$ were merely a non-supersymmetric holomorphic VOA, then
 $Z_{RR}(V)$ would be merely a (level one, meromorphic) modular function in $q$. But the  supersymmetry generator determines an odd endomorphism of $V_R$ whose square is $L_0 - \frac{c}{24}$, and so all contributions to $Z_{RR}(V)$ with $L_0 - \frac{c}{24} \neq 0$ cancel. 
 Thus $Z_{RR}(V) \in \bZ$ merely counts (with signs) the \define{Ramond-sector ground states}, i.e.\ the Ramond-sector states of conformal dimension $c/24$. This count is called the \define{(supersymmetric) index} of $V$.

When $V = V_\Lambda$ for an odd self-dual lattice $\Lambda$, the Ramond sector $(V_\Lambda)_R$ can be constructed analogously to the construction of $V_\Lambda$ above. Indeed, call a vector $\chi \in \Lambda$ \define{characteristic} if $\langle \chi,\lambda\rangle = \langle \lambda,\lambda \rangle \mod 2$ for all $\lambda \in \Lambda$. It is easy to see, using unimodularity, that a characteristic vector exists. The coset $\Lambda + \frac\chi2$ of $\Lambda$ does not depend on the choice of characteristic vector, and $V_R$ is the sum of $\Bos(\mathfrak{h})$-modules indexed by vectors in $\Lambda + \frac\chi2$. Actually, ``the'' Ramond sector of a holomorphic VOA 
 is well-defined only up to an overall fermion parity. The choice of characteristic vector $\chi$ determines this parity operator: one sets
$$ (V_\Lambda)_R = \bigoplus_{\lambda \in \Lambda} (-1)^{\langle \chi,\lambda\rangle} M_{\lambda+\frac\chi2}.$$
As in the construction of $V_\Lambda$ from the beginning of this section, ``$(-1)^{\langle \chi,\lambda\rangle}$'' indicates  the fermion parity to use for the $\Bos(\mathfrak{h})$-module $M_{\lambda+\frac\chi2}$.

\begin{proposition}\label{prop.index=index}
  Suppose $C$ is a self-dual ternary code. Then $Z_{RR}(V_{\Lambda(C)}) = \Index(C)$.
\end{proposition}

\begin{proof}
  By the remarks above, $Z_{RR}(V_{\Lambda(C)})$ merely counts states in $(V_{\Lambda(C)})_R$ of conformal dimension~$c/24$: the contributions from all other states cancel. The states of minimal conformal dimension are the $\Gamma_{\lambda + \frac\chi2}$ where $\frac12(\lambda + \frac\chi2)^2$ is as small as possible. Since $\sqrt3\bZ^c \subset \Lambda(C) \subset \frac1{\sqrt3} \bZ^c$, as our characteristic vector we may take $\chi = \sqrt3(1,1,\dots,1)$. Then $\lambda + \frac\chi2 \in \frac1{\sqrt3} (\bZ + \frac12)^c$, and so $\frac12(\lambda + \frac\chi2)^2$ is bounded below by $\frac12 \bigl( \frac1{\sqrt3}(\frac12,\frac12,\dots,\frac12)\bigr)^2 = c/{24}$, with equality only when $\lambda + \frac\chi2 \in \frac1{\sqrt3}\{ \pm \frac12 \}^c$. This, in turn, forces $\lambda + \frac\chi2 = \frac1{2\sqrt3} \iota(w)$, where $w \in C$ is a maximal codeword and $\iota : \bF_3 = \{-1,0,1\} \hookrightarrow \bZ$ is the canonical injection. So Ramond-sector ground states of $V_{\Lambda(C)}$ are in natural bijection with maximal codewords in $C$, and their signed count is precisely the index of $C$.
\end{proof}

\begin{remark} \label{remark.alternateproof}
  One can alternately prove Proposition~\ref{prop.index=index} by formula. (The following argument was pointed out to us by Noam D.\ Elkies.) Let $\eta(q) = q^{1/24} \prod_{n=1}^\infty(1-q^n)$ denote Dedekind's eta function. For an arbitrary odd unimodular lattice $\Lambda$, it follows easily from the construction of the Ramond sector of $V_\Lambda$ indicated above that
  $$ Z_{RR}(V_\Lambda) = \eta(q)^{-c} \sum_{\lambda \in \Lambda} (-1)^{\langle\lambda,\chi\rangle} q^{\frac12(\lambda + \frac\chi2)^2}. $$
  Set $\Lambda = \Lambda(C)$, choose the characteristic vector $\chi = \sqrt3(1,1,\dots,1)$, and let $\iota : \bF_3 \cong \{-1,0,1\} \hookrightarrow \bZ$ denote the natural injection. Then 
\begin{align*}
  \Lambda(C)  &  = \left\{ \sqrt 3 v + \frac1{\sqrt3}\iota(w) \st v \in \bZ^c \text{ and } w \in C\right\}.
\end{align*}
Expand each vector $v \in \bZ^c$ and each codeword $w \in C \subset \bF_3^c$ in coordinates: $v = (v_1,\dots,v_c)$, $w = (w_1,\dots,w_c)$. Note that $\langle \sqrt 3 v + \frac1{\sqrt3}\iota(w), \chi\rangle = \sum_{i=1}^c \left( v_i + \iota(w_i) \right) \bmod 2$.
Then
\begin{align*}
\hspace*{-1in}
  \sum_{\lambda \in \Lambda(C)} (-1)^{\langle\lambda,\chi\rangle} q^{\frac12(\lambda + \frac\chi2)^2} & = \sum_{w \in C} \sum_{v \in \bZ^c} (-1)^{\sum v_i + \sum\iota(w_i)} q^{\frac32\sum \left( v_i + \frac12 + \frac13\iota(w_i) \right)^2} \\
  & = \sum_{w \in C} \prod_{i=1}^c \sum_{n \in \bZ} (-1)^{n+\iota(w_i)} q^{\frac32 \left(n + \frac12 + \frac13\iota(w_i)\right)^2} \\
  & = \CWE_C\left( \sum_{n \in \bZ} (-1)^{n} q^{\frac32 \left(n + \frac12\right)^2}, \sum_{n \in \bZ} (-1)^{n+1} q^{\frac32 \left(n + \frac56\right)^2}, \sum_{n \in \bZ} (-1)^{n-1} q^{\frac32 \left(n + \frac16\right)^2}\right).
\hspace*{-1in}
\end{align*}
Here $\CWE_C$ is the complete weight enumerator mentioned in Remark~\ref{rem.CWE}. Similar formulas present the ordinary Theta function of $\Lambda(C)$ in terms of $\CWE_C$; see for example the last section of \cite{MR1794130}.

But
$$ \sum_{n \in \bZ} (-1)^{n} q^{\frac32 \left(n + \frac12\right)^2} = 0 $$
because the $n$th summand cancels with the $(1-n)$th summand. Furthermore, the Euler identity implies
$$ \sum_{n \in \bZ} (-1)^{n+1} q^{\frac32 \left(n + \frac56\right)^2} = \eta(q), \quad \sum_{n \in \bZ} (-1)^{n-1} q^{\frac32 \left(n + \frac16\right)^2} = -\eta(q). $$
All together, we find
$$ Z_{RR}(V_{\Lambda(C)}) = \eta(q)^{-c} \CWE_C(0,\eta(q),-\eta(q)) = \CWE_C(0,1,-1) = \Index(C)$$
since $\CWE_C$ is homogeneous of degree $c$.
\end{remark}

We may now prove Theorem~\ref{thm.index of a code}, which asserts that if $C$ is a self-dual ternary code of length divisible by $24$, then $\Index(C)$ is divisible by $24$.

\begin{proof}[Proof of Theorem~\ref{thm.index of a code}] 
Let $\Lambda$ be an odd unimodular lattice of rank $c$ and,  as in the proof of Proposition~\ref{prop.index=index}, choose a characteristic vector $\chi \in \Lambda$. Decompose $\Lambda = \Lambda_\ev \sqcup \Lambda_\odd$, where $\Lambda_\ev$, resp.\ $\Lambda_\odd$, is the set of vectors in $\Lambda$ of even, resp.\ odd, square-length. Consider the sets
\begin{align*}
  \Lambda^+ & = \Lambda_\ev \sqcup \left(\Lambda_\ev + \frac\chi2\right) \\
  \Lambda^- & = \Lambda_\ev \sqcup \left(\Lambda_\odd + \frac\chi2\right)
\end{align*}
When $c$ is divisible by $4$, $\Lambda_\pm$ are unimodular lattices, called the \define{neighbors} of $\Lambda$. They are even lattices when $c$ is divisible by $8$. (A different choice of characteristic vector might exchange $\Lambda_+ \leftrightarrow \Lambda_-$.) Recall that the \define{Theta function} of a lattice $\Lambda$ is
$$ \Theta_L(q) = \sum_{\lambda \in \Lambda} q^{\lambda^2/2}.$$
Then, when $c$ is divisible by $4$, we find
$$ Z_{RR}(V_\Lambda) = \eta^{-c} (\Theta_{\Lambda^+} - \Theta_{\Lambda^-}).$$

Suppose now that $C$ is a self-dual ternary code of length $c=24k$, and take $\Lambda = \Lambda(C)$. Since $\Lambda^\pm$ is even unimodular of rank $24k$,
$\Theta_{\Lambda^\pm}$ is an integral modular form of weight $12k$, and so has an expansion of the form
$$
  \Theta_{\Lambda^\pm}  = a^\pm_0 \Delta^k + a^\pm_1 \Delta^{k-1} c_4^3 + \dots + a^\pm_k c_4^{3k} 
$$
where $c_4$ denotes the weight-4 Eisenstein series and $\Delta = \eta^{24}$ is the discriminant. Since $\Index(C) = \Delta^{-k} ( \Theta_{\Lambda^+} - \Theta_{\Lambda^-})$ is a supersymmetric index, it is an integer, and so $a^+_i = a^-_i$ for $i>0$ and the index is $\Index(C) = a^+_0 - a^-_0$. But by \cite[Theorem 12.1]{MR1323986} (see also \cite[Theorem 5.10]{MR1989190}), $a^\pm_0$ is divisible by $24$.
\end{proof}

\section{Automorphisms of 3-framed SCFTs} \label{symmetry section}

The standard definition of \define{automorphism} of a self-dual ternary code $C \subset \bF_3^c$ is the following. Consider the group $2^c{:}S_c$. (The notation follows that of the ATLAS \cite{ATLAS}. In particular, the colon denotes a semidirect product, $p^n$ denotes an elementary abelian group of that order, and $S_n$ denotes the symmetric group.) It acts by signed coordinate permutations on $\bF_3^c$, and
$$ \Aut(C) = \{ g \in 2^c{:}S_c \st g(C) = C \text{ as a set}\}.$$
The signed coordinate permutation action of $2^c{:}S_c$ on $\bF_3$ lifts to actions on $\sqrt3\bZ^c$ and on $\frac1{\sqrt3}\bZ^c$, and so $\Aut(C)$ acts naturally on the lattice $\Lambda(C)$ built from $C$.

In general, the automorphism group $\Aut(\Lambda)$ of a lattice $\Lambda$ does not act on the corresponding vertex algebra $V_\Lambda$.
Indeed, the modules $M_\lambda$, or equivalently their generators $\Gamma_\lambda$, are defined only up to phase, and so the action of $\Aut(\Lambda)$ suffers from a phase ambiguity. 
 Let $\widehat{\Lambda} = \hom(\Lambda,\rU(1))$ denote the Pontryagin dual torus to $\Lambda$. Then $\Aut(V_\Lambda)$ contains a subgroup of shape $\widehat{\Lambda} .\Aut(\Lambda)$, where the dot denotes an extension which might or might not split. This subgroup, in turn, contains a subgroup of shape $\widehat{\Lambda} [2].\Aut(\Lambda)$, where $\widehat{\Lambda} [2] = \hom(\Lambda,2)$ is the 2-torsion subgroup of $\widehat{\Lambda} $; if $\Lambda$ has rank $c$, then $\widehat{\Lambda} [2] \cong 2^c$.
 (The subgroup $\widehat{\Lambda} [2].\Aut(\Lambda) \subset \widehat{\Lambda} .\Aut(\Lambda)$ is not canonical, but it is canonical up to conjugation by an element of $\widehat{\Lambda}$. Said another way, the phase ambiguity in defining $\Gamma_\lambda$ can be resolved to a sign ambiguity. That sign ambiguity cannot
 be resolved in general.)
 The subgroups $\widehat{\Lambda} [2].\Aut(\Lambda) \subset \widehat{\Lambda} .\Aut(\Lambda) \subset \Aut(V_\Lambda)$ were first studied in \cite{MR820716}, and the complete calculation of $\Aut(V_\Lambda)$ is due to \cite{MR1745258}. These results are nicely surveyed in \cite[Section 5.3]{MollerThesis}.
 
  In spite of the fact that in general the extension $\widehat{\Lambda} [2].\Aut(\Lambda) \subset \Aut(V_\Lambda)$ might not split, we claim:

\begin{theorem}\label{thm.symmetries lift}
  Let $C \subset \bF_3^c$ be a self-dual ternary code of length $c$. Then $\Aut(C) \subset \Aut(\Lambda(C))$ acts on $V_{\Lambda(C)}$ --- the extension $\widehat{\Lambda(C)} [2].\Aut(C)$ splits. The group $\Aut_{\cN=1}(V_{\Lambda(C)})$ of automorphisms of $V_{\Lambda(C)}$ preserving the $\cN=1$ supersymmetry contains a subgroup of shape ${C^*}{:}\Aut(C)$, where $C^* = \bF_3^c / C \cong 3^{c/2}$.
  The group $\Aut_{\cN=2}(V_{\Lambda(C)})$ of automorphisms preserving the $\cN=2$ structure contains a subgroup of shape $C^* {:} (\Aut(C) \cap S_c)$.
\end{theorem}

\begin{proof}
  We first observe that the signed permutation action of $2^c{:}S_c$ on $\sqrt3\bZ^c$ lifts to an action on~$V_{\sqrt3\bZ^c}$. Indeed, $V_{\sqrt3\bZ^c} = V_{\sqrt3\bZ}^{\otimes c}$ carries a permutation action by $S_c$, and the reflections $2^c$ also act since the automorphism $n \sqrt3 \leftrightarrow -n\sqrt 3$ of $\sqrt3\bZ$ lifts to an order-2 automorphism $\Gamma_{+\sqrt3} \leftrightarrow \Gamma_{-\sqrt3}$ of~$V_{\sqrt3\bZ}$. Note further that this automorphism preserves the $\cN=1$ supersymmetry generator $G = \Gamma_{+\sqrt3} + \Gamma_{-\sqrt3}$ on $V_{\sqrt3\bZ}$, and so $2^c{:}S_c \subset \Aut_{\cN=1}(V_{\sqrt3\bZ^c})$. The reflections $\Gamma_{+\sqrt3} \leftrightarrow \Gamma_{-\sqrt3}$ do not preserve the $\cN=2$ structure, but the permutations do: $S_c \subset \Aut_{\cN=2}(V_{\sqrt3\bZ^c})$.
  
  The action of the group $\widehat{\Lambda(C)}[2].\Aut(C)$ on $V_{\Lambda(C)}$ preserves (as a set) the subalgebra $V_{\sqrt3\bZ^c}$, and so we have a map
  $$ \widehat{\Lambda(C)}[2].\Aut(C) \to \widehat{\sqrt3\bZ^c}[2].(2^c{:}S_c)$$
  extending the map $\widehat{\Lambda(C)}[2] \to \widehat{\sqrt3\bZ^c}[2]$ (dual to the inclusion $\sqrt3\bZ^c \mono \Lambda(C)$) and
  covering the map $\Aut(C) \mono 2^c{:}S_c$.
  But the remarks in the previous paragraph imply that the latter extension $\widehat{\sqrt3\bZ^c}[2].(2^c{:}S_c)$ splits. Since $\sqrt3\bZ^c$ has odd index in $\Lambda(C)$, the map $\widehat{\Lambda(C)}[2] \to \widehat{\sqrt3\bZ^c}[2]$ is an isomorphism. Thus the extension $\widehat{\Lambda(C)}[2].\Aut(C)$ splits. The subgroup $\Aut(C)$ manifestly preserves the $\cN=1$ structure, and $\Aut(C) \cap S_c$ preserves the $\cN=2$ structure.
  
  Not all of the torus $\widehat{\Lambda(C)} \subset \Aut(V_{\Lambda(C)})$ preserves the supersymmetry, but  the subgroup that acts trivially on the subalgebra $V_{\sqrt3\bZ^c} \subset V_{\Lambda(C)}$ does.   This subgroup is the kernel of the map $\widehat{\Lambda(C)} \to \widehat{\sqrt3\bZ^c}$, and so is Pontryagin dual to the cokernel of the inclusion $\sqrt3\bZ^c \mono \Lambda(C)$. But there is a canonical isomorphism $\Lambda(C) / \sqrt3\bZ^c \cong C$, and its Pontryagin dual is canonically isomorphic to $C^*$.
\end{proof}

\begin{example} \label{supermoonshine continued}
Let $c=12$ and consider the  Ternary Golay code mentioned in Example~\ref{eg.scfts with c=12}, so that $V_{\Lambda(\text{Golay})} \cong V^{f\natural}$ is the ``supermoonshine'' module of \cite{MR2352133}. According to \cite{MR2352133}, $\Aut_{\cN=1}(V^{f\natural}) \cong \mathrm{Co}_1$, Conway's largest simple sporadic group. Since $\Aut(\text{Golay}) \cong 2M_{12}$, our construction makes manifest the maximal subgroup $3^6{:}2M_{12} \subset \mathrm{Co}_1$, which contains the 3-Sylow subgroup. For comparison, the construction from \cite{MR2352133} makes manifest the 2-Sylow-containing maximal subgroup of shape $2^{11}{:}M_{24}$. (The groups $M_n$ are those of Mathieu, and the ``$2$'' in $2M_{12}$ denotes the Schur cover.)

The Ternary Golay code is unique up to signed permutations, but it is not unique if one uses only unsigned coordinate permutations, and different choices produce different $\cN=2$ structures. By adjusting the signs of the coordinates, one may find a copy of the Ternary Golay code which contains the all-$1$s word $(1,1,\dots,1)$; the version given in Example~\ref{eg.codes of length 12} has this property. For this choice of Ternary Golay code, $\Aut(\text{Golay}) \cap S_c \cong M_{11}$, and so for this $\cN=2$ structure, $\Aut_{\cN=2}(V^{f\natural}) \supset 3^6{:}M_{11}$.  According to \cite{MR3373711}, a choice of $\cN=2$ structure on $V^{f\natural}$ corresponds to a choice of oriented 2-plane inside the real span of the Leech lattice, and its group of automorphisms is the subgroup of $\mathrm{Co}_1$ whose lift to $\mathrm{Co}_0=2\mathrm{Co}_1$ (the Schur cover of $\mathrm{Co}_1$) preserves this 2-plane.  There is only one 2-plane, up to isomorphism, whose automorphism group contains a group of shape $3^6{:}M_{11}$. In the notation of \cite[Chapter 10]{MR1662447}, 
it is the 2-plane spanned by a simplex of ``type 333,'' and its full automorphism group is $*333 \cong 3^6{:}M_{11}$.
\end{example}

\section{Orbifolds of 3-framed SCFTs} \label{orbifold section}

Let $V$ be a holomorphic, possibly fermionic, VOA and $G \subset \Aut(V)$ a finite group of automorphisms of $V$.  A typical question in conformal field theory is to ``gauge'' the action of~$G$ on~$V$, to produce a new holomorphic VOA $V\sslash G$ --- in the language of VOAs, the result of such a gauging procedure is the called ``twisted orbifold'' of $V$ by $G$.  One expects on physical grounds that there might be choices involved when gauging a symmetry group, and that the problem might be obstructed. Specifically, one expects to encounter an \define{'t Hooft anomaly} living in some cohomology group of $G$ such that trivializations of the anomaly correspond to choices of $V\sslash G$. 
The bosonic case is well studied in the VOA literature: the 't Hooft anomaly is an ordinary cohomology class $\alpha \in \H^3(G;\bC^\times)$~\cite{MR1003430,MR1923177,MR2730815}, and is known to exist when $G$ is solvable~\cite{CarnahanMiyamoto} and expected to exist in general.
We will need the fermionic generalization, which has been only partially explored in the VOA literature. We first need some terminology:

\begin{definition}
  Let $\cat{SVec}_\bC$ denote the symmetric monoidal category of complex super vector spaces, and $\cat{SVec}_\bC^\times$ its symmetric monoidal subcategory of $\otimes$-invertible objects and isomorphisms. The classifying space $|\cat{SVec}_\bC^\times|$ of $\cat{SVec}_\bC^\times$ has two nonvanishing homotopy groups: $\pi_0 = \bZ_2$ and $\pi_1 = \bC^\times$. Since $\cat{SVec}_\bC^\times$ is symmetric monoidal, $|\cat{SVec}_\bC^\times|$ is an infinite loop space,
  and so defines a  generalized cohomology theory that we will call \define{restricted supercohomology} and denote by $\rSH^\bullet(-)$. We choose degree conventions so that $\rSH^0(\mathrm{pt}) = \bC^\times$ and $\rSH^{-1}(\mathrm{pt}) = \bZ_2$.
  The cohomology theory $\rSH^\bullet$ is the ``supercohomology'' of \cite{GuWen2014} and the ``$E$-theory'' of \cite{MR2434259}.
   The connecting map (k-invariant) is the second Steenrod square $\Sq^2$, encoding the Koszul sign rules in $\cat{SVec}_\bC$.
  
  Let $\cat{SAlg}_\bC$ denote the symmetric monoidal bicategory of complex super algebras and super bimodules (so that equivalences in $\cat{SAlg}_\bC$ are Morita equivalences), and $\cat{SAlg}_\bC^\times$ its symmetric monoidal subbicategory of Morita-invertible algebras, invertible bimodules, and isomorphisms. Its homotopy groups are $\pi_0 = \bZ_2$, $\pi_1 = \bZ_2$, and $\pi_2 = \bC^\times$. Again its classifying space is an infinite loop space and so defines a generalized cohomology theory, called \define{extended supercohomology} $\SH^\bullet(-)$, indexed so that $\SH^0(\mathrm{pt}) = \bC^\times$. 
  It has been studied under various names. In the condensed matter literature it was introduced in \cite{WangGu2017} and studied for example in \cite[\S5.4]{GJF2017}.
\end{definition}

By construction, for any space $X$ there are long exact sequences 
\begin{gather*}
 \dots \to \H^\bullet(X;\bC^\times) \to \rSH^\bullet(X) \to \H^{\bullet-1}(X;\bZ_2) \to \H^{\bullet+1}(X) \to \dots,\\
 \dots \to \rSH^\bullet(X) \to \SH^\bullet(X) \to \H^{\bullet-2}(X;\bZ_2) \to \rSH^{\bullet+1}(X) \to \dots
\end{gather*}
corresponding to the inclusions $\rB \bC^\times \mono \cat{SVec}_\bC^\times$ and $\rB \cat{SVec}_\bC^\times \mono \cat{SAlg}_\bC^\times$, where the letter $\rB$ denotes turning a commutative algebra (resp.\ symmetric monoidal category) into a symmetric monoidal category (resp.\ bicategory) with one object.
Given a class in $\SH^\bullet(X)$, its image in $\H^{\bullet-2}(X;\bZ_2)$ is called its \define{Majorana layer}. Given a class in $\rSH^\bullet(X)$, its image in $\H^{\bullet-1}(X;\bZ_2)$ is called its \define{Gu--Wen layer}. The connecting maps in the above long exact sequences are stable cohomology operations, and so vanish in very low degrees. In particular, for $\bullet=3$, we have inclusions
$$ \H^3(X;\bC^\times) \mono \rSH^3(X) \mono \SH^3(X). $$

When $G$ is a finite group and $\rH^\bullet(-)$ a cohomology theory, we will write $\rH^\bullet_\gp(G)$ in place of $\rH^\bullet(\rB G)$. We will denote reduced cohomology by $\tilde{\rH}$, $\tilde{\SH}$, etc. Recall that a VOA is \define{regular} if its category of admissible modules is finite semisimple.

\begin{theorem}\label{thm.anomaly}
  Let $V$ be a holomorphic, possibly fermionic, VOA and $G$ a finite group acting faithfully on (the NS sector of) $V$. Assume that the $G$-fixed sub-VOA $V^G$ is regular. Then there is a well-defined 't~Hooft anomaly $\alpha \in \tilde{\SH}^3_\gp(G)$. Each trivialization of $\alpha$ determines an orbifold VOA $V \sslash G$. In particular, there are $\tilde{\SH}^2_\gp(G)$-many orbifolds.
  
  When $V$ is bosonic, the 't Hooft anomaly $\alpha$ lives in $\tilde{\H}^3_\gp(G;\bC^\times) \subset \tilde{\SH}^3_\gp(G)$. Trivializations in $\tilde{\H}^2_\gp(G;\bC^\times) \subset \tilde{\SH}^2_\gp(G)$ give bosonic orbifolds.
\end{theorem}

Note that if $G$ is solvable, then regularity of $V^G$ follows from the main result of~\cite{CarnahanMiyamoto}. (That paper discusses the bosonic case, but the fermionic case follows immediately by passing to the bosonic subalgebra of $V$.)

\begin{proof}
The bosonic statement is the main result of \cite{MR1923177}. The proof of the fermionic version is essentially the same,  replacing everywhere the word ``category'' with ``supercategory.'' (A $\bC$-linear \define{supercategory} is a category enriched and tensored over $\cat{SVec}_\bC$.) For each $g \in G$, consider the supercategory $\cat{Rep}(V,g)$ of $g$-twisted $V$-modules. It is nonempty, and regularity of $V^G$ provides a fusion product $\otimes: \cat{Rep}(V,g) \times \cat{Rep}(V,g') \to \cat{Rep}(V,gg')$. Thus the direct sum $\cC = \bigoplus_{g\in G} \cat{Rep}(V,g)$ is a $G$-graded ``superfusion category.'' Its neutral component $\cat{Rep}(V,1)$ is a copy of $\cat{SVec}_\bC$ since $V$ is holomorphic. The tensor product of nonzero objects in a (super) fusion category never vanishes. (Indeed, let $X,Y \in \cC$ be nonzero. Then $\operatorname{id}_Y \neq 0 \in \hom(Y,Y)$, which by adjunction says that the coevaluation map $\operatorname{coev}_Y : \mathbf{1}_\cC \to Y \otimes Y^*$ is nonzero, hence an injection since the unit object $\mathbf{1}_\cC$ is simple. But then $\operatorname{id}_X \otimes \operatorname{coev}_Y : X \to X \otimes Y \otimes Y^*$ is an injection, and so $X \otimes Y$ cannot vanish.) Fix $g$, and suppose $X,Y \in \cat{Rep}(V,g)$. Then $X \otimes Y^* \in \cat{Rep}(V,1) = \cat{SVec}_\bC$ is nonzero, and so $\hom(Y,X) = \hom(\mathbf{1}_\cC,X \otimes Y^*)$ is a nonzero super vector space. It follows that for each $g$, $\cat{Rep}(V,g)$ has only one (up to possibly odd isomorphism) simple object,
  with endomorphism algebra either $\bC$ or $\operatorname{Cliff}(1,\bC)$.
  
  There is a natural ``Deligne'' tensor product $\boxtimes$ of supercategories. Consider the assignment $g \mapsto \cat{Rep}(V,g)$. The previous paragraph shows that $\cat{Rep}(V,g)$ is $\boxtimes$-invertible, and the fusion $\otimes: \cat{Rep}(V,g) \times \cat{Rep}(V,g') \to \cat{Rep}(V,gg')$ turns the assignment $g \mapsto \cat{Rep}(V,g)$ into a monoidal functor
  $$ G \to \{\text{$\boxtimes$-invertible supercategories}\},$$
  where $G$ is considered a monoidal bicategory with only identity 1- and 2-morphisms. The functor that to each superalgebra assigns its supercategory of modules provides an equivalence $$\cat{SAlg}_\bC^\times \simeq \{\text{$\boxtimes$-invertible supercategories}\}$$ of (symmetric) monoidal bicategories. All together, we have constructed a monoidal map $\alpha : G \to \cat{SAlg}_\bC^\times$, 
  which is precisely the data of a class $\alpha \in \SH^3_\gp(G)$.
  
  Following \cite{MR1923177}, we can identify $\cat{Rep}(V^G)$ with the Drinfeld center of $\cC = \bigoplus_{g\in G} \cat{Rep}(V,g)$, understood of course in the super sense. Trivializations of $\alpha$ determine algebra objects in $\cat{Rep}(V^G)$ with certain nice properties (in particular, they are ``Lagrangian''), and hence  extensions of $V^G$ to holomorphic VOAs. These extensions are the orbifolds  $V\sslash G$.
\end{proof}

\begin{example}\label{eg.Klein4}
  Suppose $V$ is a holomorphic VOA with a nonanomalous action by the Klein-4 group~$G=\bZ_2^2$. We will work out all possible orbifolds of $V$ by subgroups of $G$. 
  
  First, suppose that $V$ is bosonic --- what are its bosonic orbifolds? There are three copies of $\bZ_2 \mono G$, each nonanomalous, and so we find three VOAs $W_i = V \sslash \bZ_2$. There are also two ways to orbifold $V$ by all of $G$, since $\H^2_\gp(G;\bC^\times) \cong \bZ_2$. Anticipating the answer, we rename $V = V_1$, and call the two orbifolds $V_1 \sslash G$ by the names $V_2$ and $V_3$. One can form $V_2$ from $V_1$ in two steps. First, orbifold one of the $\bZ_2$s inside $G$, say the first one $\langle g \rangle \subset G$, resulting in the VOA $W_1$. Then $W_1$ has a new $\bZ_2$ acting on it by changing the sign of the twist field. It also has an ``old'' $\bZ_2 = G / \langle g \rangle$ symmetry. Actually, to define this ``old'' $\bZ_2$ action requires the data of the trivialization of the anomaly $\alpha$: $W_1$ has two canonical ``old'' $\bZ_2$ actions, differing by the action of the ``new'' $\bZ_2$. After choosing the trivialization of $\alpha$, we can build $V_2$ by orbifolding $W_1$ by this old $\bZ_2$ action. If we had used the other trivialization of $\alpha$, hence the other ``old'' $\bZ_2$, we would land not at $V_2$ but at $V_3$. If instead we had orbifolded $W_1$ by the ``new'' $\bZ_2$ that changed the sign of the twist field, we would return to $V_1$. All together we find that our six VOAs $V_1,V_2,V_3,W_1,W_2,W_3$ form a $K_{3,3}$ graph, where the edges denote possible $\bZ_2$-orbifolds:
  $$ 
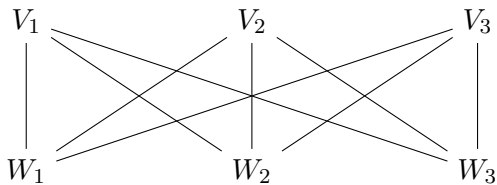
\begin{tikzpicture}
    \path [anchor=base] (0,0) node (V1) {$V_1$}
    (3,0) node (V2) {$V_2$}
    (6,0) node (V3) {$V_3$}
    (0,-2) node (W1) {$W_1$}
    (3,-2) node (W2) {$W_2$}
    (6,-2) node (W3) {$W_3$};
    \draw (V1) -- (W1);
    \draw (V1) -- (W2);
    \draw (V1) -- (W3);
    \draw (V2) -- (W1);
    \draw (V2) -- (W2);
    \draw (V2) -- (W3);
    \draw (V3) -- (W1);
    \draw (V3) -- (W2);
    \draw (V3) -- (W3);
  \end{tikzpicture}$$
  
  Consider now to the situation that $V=V_1$ is still bosonic, but  allow fermionic orbifolds. 
  For each $\bZ_2 \subset G$, there is a unique bosonic orbifold $V\sslash\bZ_2$, since $\H^2_\gp(\bZ_2;\bC^\times) = 0$, but there is also a fermionic one, since $\tilde{\SH}_\gp^2(\bZ_2) = \bZ_2$. Note that since $V$ is bosonic and holomorphic, its central charge, hence the central charge of all of its orbifolds, is divisible by $8$. Now suppose $F$ is any actually-fermionic holomorphic VOA of central charge divisible by $8$ --- by ``actually fermionic'' we mean that $F$ is not bosonic. Then $F$ has two \define{bosonic neighbors}, defined as the two bosonic holomorphic VOAs containing the even subalgebra of $F$, or equivalently the two results of gauging the fermion parity operator $(-1)^f$ on $F$. (Only when the central charge is divisible by $8$ is $(-1)^f$ is nonanomalous with $F \sslash  (-1)^f$  bosonic.) These two neighbors $V,W$ of $F$ are related by a (bosonic) $\bZ_2$-orbifold. In this way actually-fermionic holomorphic VOAs of central charge divisible by $8$ are identified with pairs of bosonic VOAs related by $\bZ_2$ orbifold. This identification is a VOA analog of the ``even neighbors'' of an odd unimodular lattice of rank divisible by $8$. 
  
  Returning to our question, we find that the actually-fermionic orbifolds of $V = V_1$ by subgroups of $G = \bZ_2^2$ are in bijection with the edges of the above $K_{3,3}$ graph. For $i,j \in \{1,2,3\}$, we will let $F_{ij}$ denote the fermionic VOA whose neighbors are $V_i$ and $W_j$. The orbifolds of $V_1$ by a single $\bZ_2 \subset G$ are the $W_j$s and the $F_{1j}$s. (There are three $\bZ_2$s, and each gives $\tilde{\SH}^2_\gp(\bZ_2) = \bZ_2$ many orbifolds.) The remaining $V_2,V_3,F_{2j},F_{3j}$ are the results of orbifolding $V_1$ by all of $G = \bZ_2^2$. To check that we have not missed any, note that there should be precisely $|\tilde{\SH}^2_\gp(\bZ_2^2)| = |\H^2_\gp(\bZ_2^2;\bC^\times) . \H^1_\gp(\bZ_2^2;\bZ_2)| = 8$ such orbifolds.
  
  Just as in the bosonic case, the fermionic orbifolds $F_{ij}$ each come with an action of a Klein-4 group $\bZ_2^2$, which now contains the fermion parity operator $(-1)^f$. Each of the other $\bZ_2$s provides two orbifolds, which are themselves related by a $\bZ_2$-orbifold: in the fermionic world, $\bZ_2$-orbifolds always come in trios (for example, an actually-fermionic VOA together with its even neighbors). All together, we find that the above $K_{3,3}$ graph expands to the following:
  $$ 
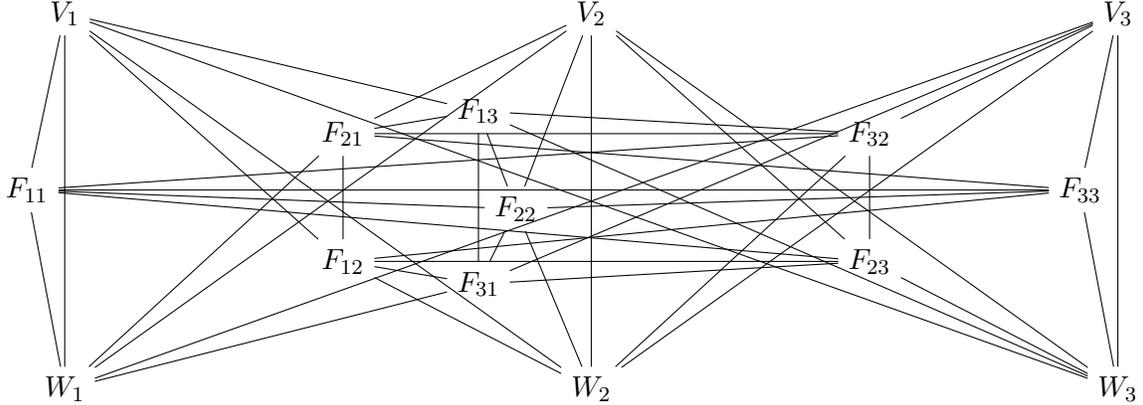
\begin{tikzpicture}
    \path [anchor=base] (0,-.5) node (V1) {$V_1$}
    (7,-.5) node (V2) {$V_2$}
    (14,-.5) node (V3) {$V_3$}
    (0,-5.5) node (W1) {$W_1$}
    (7,-5.5) node (W2) {$W_2$}
    (14,-5.5) node (W3) {$W_3$};
    \path (-.5,-2.75) node (F11) {$F_{11}$}; \draw (V1) -- (F11) -- (W1) -- (V1);
    \path (3.7,-3.7) node (F12) {$F_{12}$}; \draw (V1) -- (F12) -- (W2) -- (V1);
    \path (5.5,-1.7) node (F13) {$F_{13}$}; \draw (V1) -- (F13) -- (W3) -- (V1);
    \path (3.7,-2) node (F21) {$F_{21}$}; \draw (V2) -- (F21) -- (W1) -- (V2);
    \path (6,-3) node (F22) {$F_{22}$}; \draw (V2) -- (F22) -- (W2) -- (V2);
    \path (10.7,-3.7) node (F23) {$F_{23}$}; \draw (V2) -- (F23) -- (W3) -- (V2);
    \path (5.5,-4) node (F31) {$F_{31}$}; \draw (V3) -- (F31) -- (W1) -- (V3);
    \path (10.7,-2) node (F32) {$F_{32}$}; \draw (V3) -- (F32) -- (W2) -- (V3);
    \path (13.5,-2.75) node (F33) {$F_{33}$}; \draw (V3) -- (F33) -- (W3) -- (V3);
    \draw (F11) -- (F22) -- (F33) -- (F11);
    \draw (F12) -- (F21) -- (F33) -- (F12);
    \draw (F11) -- (F23) -- (F32) -- (F11);
    \draw (F13) -- (F22) -- (F31) -- (F13);
    \draw (F12) -- (F23) -- (F31) -- (F12);
    \draw (F13) -- (F21) -- (F32) -- (F13);
  \end{tikzpicture}$$
  
  Finally, in the case where there is no bosonic theory in sight, one finds the same graph of possible orbifolds, with the only difference being that there is no way to distinguish which vertices are ``$V$'', ``$F$'', or ``$W$''. This graph has a more symmetrical description. (We thank David Treumann for helping to understand this point.) The vertices correspond to Lagrangian 2-planes inside symplectic~$\bF_2^4$ (equivalently Lagrangian algebras in $\cat{Rep}(V^G)$). The edges correspond to pairs of planes intersecting in a line. Every isotropic line is contained in precisely three Lagrangian planes, hence the triangles. The symmetry group of the graph is the symplectic group $\operatorname{Sp}(4,\bF_2) \cong S_6$.
\end{example}

The construction of $\alpha \in \SH^3_\gp(G)$ in the proof of Theorem~\ref{thm.anomaly} identifies its layers as follows. Choose, for each $g\in G$, a simple object $V(g) \in \cat{Rep}(V,g)$. The object $V(g)$ is ``the'' $g$-twisted sector --- the word ``the'' is in quotes because $V(g)$ is unique only up to possibly-odd isomorphism. 
The Majorana layer $\alpha^{(1)} \in \H^1_\gp(G;\bZ_2)$ records whether $V(g)$ is ``ordinary'' with $\End(V(g)) \cong \bC$  or ``Majorana'' with $\End(V(g)) \cong \operatorname{Cliff}(1,\bC)$. Suppose that the Majorana layer vanishes, and let $\alpha^{(2)} \in \H^2_\gp(G;\bZ_2)$ denote the Gu--Wen layer of $\alpha$. Then $\alpha^{(2)}(g_1,g_2)$ records whether the isomorphism $V(g_1) \otimes V(g_2) \cong V(g_1g_2)$ is even or odd. Finally, the ``ordinary cohomology'' layer of $\alpha$ provides the associator on $\cC$, just as in the bosonic case.

There is another description of the Majorana and Gu--Wen layers of $\alpha$ that is more useful for computation. (One can show the two descriptions agree by studying the ``bosonic shadow'' of the super fusion category $\bigoplus_{g\in G} \cat{Rep}(V,g)$ of $G$-twisted $V$-modules; see \cite{BGK2016}.) Consider the action of $g\in G$ on ``the'' Ramond sector $V_R$ of $V$. Again the word ``the'' is in quotes because $V_R$ is defined only up to possibly-odd isomorphism. In particular, the action of $g$ on $V_R$ might be even or odd, determined by the Majorana layer $\alpha^{(1)}(g)$. Suppose that the Majorana layer vanishes, so that the anomaly $\alpha$ lives in restricted supercohomology. Then $G$ acts by even automorphisms of $V_R$, but the action may still be projective (since $V_R$ is determined only up to isomorphism). The projectivity of the action is precisely the Gu--Wen layer $\alpha^{(2)} \in \H^2_\gp(G;\bZ_2)$.

In this paper
 we care most about the case when $G = \bZ_m$ is a cyclic group.
The abelian group structure of extended supercohomology can be hard to understand --- formulas for it in~\cite{BGK2016} require the quaternion group $Q_8$. Restricted supercohomology is easier. In particular, if $G = \bZ_m = \langle g \rangle$ is a cyclic group of order $m$ with generator $g$, then 
$$ \rSH^3_\gp(\bZ_m) \cong \begin{cases} \bZ_m, & m \text{ odd,} \\ \bZ_{2m}, & m \text{ even,} \end{cases}$$
generated by the 't Hooft anomaly of the action of $\bZ_m$ on $\Fer(2)$ in which $g$ acts by $\frac{2\pi}m$-rotation of the two fermions.

Suppose $G = \bZ_m = \langle g \rangle$ acts on $V$ with trivial Majorana layer. Consider the \define{character} of $g$, defined as
$$Z_{NS,NS}^{1,g}(V) = \tr_V(g q^{L_0 - c/24}).$$
Since $V$ is holomorphic, this character will have good modularity properties. In particular, it will be a meromorphic level-$m$ modular function, perhaps with multiplier, when $m$ is even. When $m$ is odd, the spin structures get rearranged by the modular transformation $T^m$, and so $Z_{NS,NS}^{1,g}(V)$ will have level $2m$. In both cases, if the action were nonanomalous, the multiplier would agree with the multiplier of $Z_{NS,NS}^{1,1}(V) = \tr_V(q^{L_0 - c/24})$. 
By studying the $\frac{2\pi}m$-rotation of two free fermions, one finds that when $m$ is odd and 
$\bZ_m = \langle g \rangle$ acts on $V$ with anomaly $\alpha \in \rSH^3(\bZ_m) = \bZ_m$, then
$ST^{2m}S^{-1}$ will act on $Z_{NS,NS}^{1,g}(V)/Z_{NS,NS}^{1,1}(V)$ 
with eigenvalue $\exp(\alpha\frac{2\pi i}m)$. If $m$ is even and $\bZ_m = \langle g \rangle$ acts on $V$ with anomaly $\alpha \in \rSH^3(\bZ_m) = \bZ_{2m}$, then 
$ST^{m}S^{-1}$ will act on $Z_{NS,NS}^{1,g}(V)/Z_{NS,NS}^{1,1}(V)$ 
with eigenvalue $\exp(\alpha\frac{2\pi i}{2m})$.
Thus, when $G$ is cyclic and the Majorana layer vanishes, the multiplier fully determines the 't Hooft anomaly (compare \cite[Section 3]{GPRV}).

\begin{example}\label{eg.anomaly of a lattice automorphism}
  Let $\Lambda$ be a (possibly odd) unimodular lattice, and $g \in \Aut(\Lambda)$ an automorphism of order $m$. Choose a lift $\tilde{g}$ of $g$ to $\widehat{\Lambda}.\Aut(\Lambda) \subset\Aut(V_\Lambda)$. Choose also a characteristic vector $\chi \in \Lambda$, so that the Ramond sector is built from the states $\Gamma_{\lambda + \frac\chi2}$ for $\lambda \in \Lambda$. Then $\tilde{g}$ acts on the Ramond sector by $\tilde{g}(\Gamma_{\lambda + \frac\chi2}) \propto \Gamma_{g(\lambda) + \frac{g(\chi)}2}$, where the proportionality factor has not been determined. In particular, $\tilde{g}$ is even or odd depending only on the relative parity of the vectors $\lambda + \frac\chi2$ and $g(\lambda) + \frac{g(\chi)}2$. Since $\chi$ is characteristic, this relative parity is
  $$ \left\langle \chi, \left(\lambda + \frac\chi2\right) - \left(g(\lambda) + \frac{g(\chi)}2\right)\right\rangle = \langle \chi, (1-g)\lambda\rangle + \left\langle \chi, (1-g)\frac\chi2\right\rangle = \frac12 \langle \chi, (1-g)\chi\rangle \mod 2.$$
  The second equality follows from the fact that $\langle \chi,g\lambda\rangle = (g\lambda)^2 = \lambda^2 = \langle \chi,\lambda\rangle \bmod 2$, for any $\lambda \in \Lambda$.
  In particular,
  $\langle \chi, (1-g)\chi\rangle \in 2\bZ$, and
  we find that the Majorana layer of the 't Hooft anomaly of the $\tilde{g}$ action vanishes if and only if 
$\langle \chi, (1-g)\chi\rangle \in 4\bZ$.

  Suppose that this Majorana layer does vanish, and suppose further that $g$ has no (nonzero) fixed points. Then all lifts $\tilde{g}$ of $g$ are conjugate.
  Since conjugate automorphisms have the same 't Hooft anomaly, we will drop the tilde, writing $g \in \Aut(V_\Lambda)$ for any lift of $g \in \Lambda$.
  Let us assume that this lift still has order $m$. This is automatic when $m$ is odd; when $m$ is even, it happens if and only if $\langle \lambda,(1 - g^{m/2})\lambda\rangle \in 2\bZ$ for all $\lambda$. 
  
  Since the action of $g$ on $\Lambda$ has no fixed points, the character of ${g}$ is easy to compute. Since $g$ is a lattice automorphism, its characteristic polynomial factors as $$ \det(g-x) = \prod_{k|m} (1-x^k)^{d_k} $$
for some integers $d_k$. The corresponding formal expression $\prod k^{d_k}$ is the \define{Frame shape} of $g$, introduced by Frame in \cite{MR0269751}.
Define the corresponding eta product to be
$$ \eta_g(q) = \prod_{k|m} \eta(q^k)^{d_k}. $$
  Then standard formulas give
  $$ Z_{NS,NS}^{1,g} = \frac1{\eta_g(q)}.$$
  Noting that $\eta$ itself has a multiplier of $\exp(\frac{2\pi i}{24})$ under the action of $T$, we find that $ST^{m}S^{-1}$ (or $ST^{2m}S^{-1}$, if $m$ is odd) acts on $Z_{NS,NS}^{1,g}/Z_{NS,NS}^{1,1}$ with eigenvalue
  $$ \exp\left( \frac{2\pi i}{24} \sum_{k|m} d_k \frac m k (k^2-1)\right). $$
  
  It is a basic fact that if $k$ is coprime to $6$, then $k^2-1$ is divisible by $24$. (Conway and Norton in~\cite{MR554399} call this fact ``the defining property of the number $24$.'') In particular, lattice automorphisms of order coprime to $6$ always have vanishing 't Hooft anomaly. Assuming the Majorana layer vanishes, the 't Hooft anomaly of an arbitrary fixed-point-free lattice automorphism always has order dividing $24$, and is determined by the Frame shape of the automorphism: $\alpha$ ``$=$'' $\sum_{k|m} d_k \frac m k (k^2-1) \bmod 24$.
  
    If $g$ had fixed a sublattice $\Lambda^g \subset \Lambda$, then the character would have a numerator of the form $\sum_{\lambda \in \Lambda^g} \phi(\lambda) q^{\lambda^2/2}$ for some phases $\phi(\lambda)$, and the lift of $g$ would not be unique up to conjugation. There is always a ``standard'' lift \cite{MR820716}, for which these phases are trivial. In this case the 't Hooft anomaly is again determined the the Frame shape (provided the Majorana layer vanishes); compare \cite[Chapter 5]{MollerThesis}. We warn the reader, however, that, when $\Lambda$ is odd and $m$ is even, the most natural lift might not be the ``standard'' one, and encourage the interested reader to consider the case of $g = \bigl( \begin{smallmatrix} 0 & 1 \\ 1 & 0 \end{smallmatrix}\bigr)$ acting on the $\bZ^2$ lattice.
\end{example}

\begin{example}\label{eg.bandb}
  There are precisely two self-dual ternary codes of length $c=24$ with minimal Hamming weight $9$ \cite{MR633414}: the quadratic residue code $Q_{23}$ and Pless's code $P_{11}$ from \cite{MR0245455}. Letting $C$ be either of these codes, the corresponding lattice $\Lambda(C)$ will be a rank-$24$ lattice whose shortest vectors have square-length $3$. There is a unique such lattice, the ``odd Leech lattice'' $\OL$ discovered in \cite{MR0010153}. The automorphism group of $\OL$ has shape $2^{12}{:}M_{24}$; the codes $Q_{23}$ and $P_{11}$ make visible the subgroups $\operatorname{SL}_2(\bF_{23})$ and $\operatorname{SL}_2(\bF_{11})$ respectively.
  
  Each of these codes $C$ equips $V_{\OL} = V_{\Lambda(C)}$ with an $\cN=1$ supersymmetry. These two supersymmetries are not related by an automorphism of $V_{\OL}$. Indeed, since the shortest vectors in $\OL$ have square-length $3$, the connected component of $\Aut(V_{\OL})$ is
  merely the dual torus 
  $\widehat{\OL}$, and so the  action of  $\widehat{\OL}$ on $V_{\OL}$ is canonical (and not just canonical up to isomorphism). Decompose the supersymmetry generator coming from the self-dual ternary code~$C$ into $\widehat{\OL}$-eigenvectors. Its ``support'' (i.e.\ those eigenvalues appearing with nonzero coefficient in the decomposition) generates the sublattice $\sqrt3\bZ^{24} \subset \Lambda(C) = \OL$. Since the only automorphisms of $\sqrt3\bZ^{24}$ are signed coordinate permutations, the embedding $\sqrt3\bZ^{24} \mono \OL$, which was determined by the supersymmetry generator,  determines $C$ up to isomorphism.

  Consider the automorphism $g : \lambda \mapsto -\lambda$ of $\OL$ of order $m=2$. As with all lattices, $g$ lifts to an order-2 automorphism of $V_{\OL}$, unique up to conjugation.  
   As with all lattices built from self-dual ternary codes, we may choose the characteristic vector $\chi = \sqrt3(1,1,\dots,1) \in \sqrt3\bZ^{24} \subset \Lambda(C)$, and calculate $\langle \chi, (1-g)\chi\rangle = 2\chi^2 = 6c = 144 \in 4\bZ$. So the Majorana layer of the 't Hooft anomaly of $g$ vanishes. The full anomaly may be calculated as in Example~\ref{eg.anomaly of a lattice automorphism}: the Frame shape of $g$ is $1^{-24}2^{24}$, and $\sum_{k|m} d_k \frac m k (k^2-1) = (-12)\frac21 (1^2-1) + (24)\frac22 (2^2-1) = 72 = 0 \bmod 24$. Thus the action of $g$ on $V_{\OL}$ is nonanomalous.
   
   We claim that both orbifolds $V_{\OL} \sslash  (\lambda \mapsto -\lambda)$ are isomorphic to the ``beauty and the beast'' odd VOA from \cite{MR968697}, which we will call $V^\sharp$ (it is not given a name in \cite{MR968697}). Recall  the usual construction of the Moonshine VOA $V^\natural$ from \cite{MR996026}: $V^\natural = V_{\Leech} \sslash  (\lambda \mapsto -\lambda)$, where implicitly the bosonic orbifold is chosen. One can instead choose the fermionic orbifold, with bosonic neighbors $V_\Leech$ and $V^\natural$. That fermionic orbifold is precisely $V^\sharp$.
   
   To explain the isomorphism $V_{\OL} \sslash  (\lambda \mapsto -\lambda) \cong V^\sharp$, let us first recall some of the standard Moonshine story from \cite{MR996026}. The presentation $V^\natural = V_{\Leech} \sslash  (\lambda \mapsto -\lambda)$ makes visible the maximal subgroup of shape $2^{1+24}.\mathrm{Co}_1$ of the Monster group~$\bM$; to see the whole Monster group, one must make visible one further automorphism, and this is done in \cite{MR996026} by starting not with the Leech lattice but with the Niemeier lattice $\Niem$ with root system $A_1^{24}$. The connected component of the automorphism group of $V_{\Niem}$ is isomorphic to the quotient $\SU(2)^{24} / 2^{12}$, where the normal subgroup $2^{12} \subset \mathrm{Center}(\SU(2)^{24}) = 2^{24}$ is a copy of the binary Golay code, and in particular contains the element $(-I,-I,\dots,-I)$, where $I \in \SU(2)$ is the identity matrix. Consider the Klein-4 subgroup of $\SU(2)^{24} / 2^{12}$ generated by the elements $g = (S,\dots,S)$ and $h = (H,\dots,H)$, where $S = \bigl(\begin{smallmatrix} 0 & 1 \\ -1 & 0 \end{smallmatrix}\bigr)$ and $H = \bigl( \begin{smallmatrix} i & 0 \\ 0 & -i \end{smallmatrix}\bigr)$. The orbifold $V_\Niem \sslash  \langle h\rangle$ is manifestly isomorphic to $V_\Leech$, since $h$ is the lattice momentum vector picking out the lattice neighborship relating $\Niem$ and $\Leech$, whereas $g$ is a lift of the lattice involution $\lambda \mapsto -\lambda$. For either trivialization of the anomaly, the action of $\langle g \rangle$ on $V_\Leech \cong V_\Niem \sslash  \langle h\rangle$ is a lift of $\lambda \mapsto -\lambda$. Thus we find that, for either trivialization, $V_\Niem \sslash  \langle g,h\rangle \cong V_\Leech \sslash  \langle g \rangle \cong V^\natural$.
   But $g$ and $h$ are conjugate in $\SU(2)^{24} / 2^{12}$, and this symmetry survives to $V^\natural$, providing the extra automorphism desired. 
   
   In terms of the bosonic $K_{3,3}$ diagram from Example~\ref{eg.Klein4}, we have $V_1 \cong V_\Niem$, $W_1 \cong W_2 \cong W_3 \cong V_\Leech$, and $V_2 \cong V_3 \cong V^\natural$. From this we can fill in the fermionic theories: $F_{1j} \cong V_\OL$ for all $j$, and $F_{2j} \cong F_{3j} \cong V^\sharp$. Identifying $F_{11} = V_{\OL}$ with the fermionic orbifold $V_{\Niem} \sslash  \langle h \rangle$, we find that the other $\bZ_2$ actions on $F_{11}$ --- the groups $\langle g \rangle$ and $\langle gh \rangle$ --- both lift the lattice involution $\lambda \mapsto -\lambda$. Inspecting the fermionic diagram from Example~\ref{eg.Klein4}, we find that the possible orbifolds $V_\OL \sslash  (\lambda \mapsto -\lambda)$ are the theories called $F_{ij}$ with $i,j\in\{2,3\}$, which are all isomorphic to $V^\sharp$.
   
   The main result of the paper \cite{MR968697} is that $V^\sharp$ admits $\cN=1$ supersymmetry. Their proof is nonconstructive (and we were unable to follow certain important steps in it). By instead realizing $V^\sharp$ as $V_{\Lambda(C)} \sslash  (\lambda \mapsto -\lambda)$ for a self-dual ternary code $C$, we have made the supersymmetry explicit. Indeed, the two different codes $P_{11}$ and $Q_{23}$ provide, a priori, two different $\cN=1$ supersymmetries to $V^\sharp$. We do not know if they are in fact isomorphic.
\end{example}

\section{Topological modular forms}\label{TMF section}

Topological Modular Forms (TMF) are a fairly mysterious object from stable homotopy theory that refines the usual ring of modular forms (see for example \cite{MR1989190,MR3223024}). TMF is a ``chromatic height two'' analog of oriented K-theory: ``chromatic height'' in stable homotopy theory is roughly the same as ``category number'' and so one expects that, whereas the K-theory of a manifold measures its 1-category of vector bundles, the TMF of that manifold measures some 2-category. TMF owes its origins to Witten's work connecting supersymmetric string theory to the K-theory of loop spaces \cite{MR885560,MR970288}. Building on those ideas, Stolz and Teichner proposed a conjectural geometric description of TMF~\cite{MR2079378,MR2742432}. Translated into more physical language, their description is:

\begin{conjecture}[Stolz--Teichner] \label{conj.ST}
  Let $X$ be a manifold. The degree-$n$ TMF of $X$ is
$$ \TMF^n(X) = \pi_0 \left\{ \begin{array}{c}
  \cN=(0,1)\text{ boundary conditions for the $n$th power} \\ \text{of the $c=\frac12$ invertible fermionic (2+1)d TFT} \\ \text{which couple to a background scalar field valued in }X \end{array}
\right\}. $$
\end{conjecture}

By ``the $n$th power of the $c=\frac12$ invertible fermionic (2+1)d TFT,'' we mean the invertible topological order admitting $n$ antichiral Majorana fermions as a boundary condition. Take any $\cN=(0,1)$ boundary condition for this bulk theory. Up to minor ambiguities (see Remark~\ref{remark.ambiguity}), we can identify the boundary condition with an $\cN=(0,1)$ SQFT with gravitational anomaly $n/2$.  Taking $X$ to be a point, we have:

\begin{conjecture}[Stolz--Teichner, simplified] \label{conj.STsimplified}
  $$\pi_n \TMF = \pi_0 \{\cN=(0,1) \text{ SQFTs with gravitational anomaly } n/2\}.$$
\end{conjecture}

\begin{remark} \label{remark.ambiguity}
Recall that for an SCFT, the gravitational anomaly is $c_R - c_L$; the notion of ``gravitational anomaly'' makes sense even when the SQFT is not conformal. The relationship between gravitational anomalies and boundary theories is explored, among other places, in \cite{FreedTeleman2012}.

Conjecture~\ref{conj.STsimplified} should not be taken too literally because of a minor sign ambiguity stemming, ultimately, from the choice of fermion parity of ``the'' Ramond sector. This leads to a sign ambiguity in defining the Ramond-Ramond partition function. Changing that sign corresponds to stacking the theory with an invertible fermionic (1+1)d TFT.

There is an alternate way to relate the Conjectures. By employing the reference Majorana fermions boundary condition, one can map the $\cN=(0,1)$ boundary condition to an $\cN=(0,1)$ SQFT with no gravitational anomaly and a decoupled subsector consisting of $n$ left-moving Majorana fermions, acted upon trivially by the $\cN=(0,1)$ supercharge. Conversely, one can rephrase Conjecture~\ref{conj.ST} to avoid any mention of (2+1)d TFTs by instead using the moduli space of nonanomalous SQFTs equipped with such a subsector. The coset of the subsector is the gravitationally-anomalous SQFT in Conjecture~\ref{conj.STsimplified}.
\end{remark}

A basic property of topological modular forms is that there is a map from TMF to ordinary cohomology with coefficients in the ring $\MF$ of ordinary integral modular forms. (The homotopy theorists' convention is to think of $\MF$ as a graded-commutative ring concentrated entirely in even degrees, so that weight-$c$ modular forms are in homotopical degree $n=2c$.)
When the primes $2$ and $3$ are inverted, this map $\TMF \to \MF$ is an isomorphism, but it has interesting kernel and cokernel at the primes $2$ and $3$. The map has been computed on homotopy. In particular, the image of $\pi_n \TMF \to \MF_{n/2}$ is fully understood \cite[Proposition 4.6]{MR1989190}. It is not a surjection. Its most interesting feature is that $m \Delta^k$ is in the image of $\pi_n \TMF \to \MF_{n/2}$ when and only when~$24$ divides~$mk$.

In terms of Conjecture~\ref{conj.STsimplified}, the map $\TMF \to \MF$ should take an $\cN=(0,1)$ SQFT~$V$ to its Ramond-Ramond partition function $Z_{RR}(V)$, multiplied by $\eta^n$. The supersymmetry plays two roles: it protects $Z_{RR}(V)$ to be invariant under deformations, and in particular under RG flow; it makes $Z_{RR}(V)$, which for a general QFT would be a function of both~$q$ and~$\bar{q}$, into a function of~$q$ alone. As a special case, Conjecture~\ref{conj.ST} predicts:

\begin{conjecture}
  Each holomorphic VOA $V$ of central charge $c$ equipped with $\cN=1$ supersymmetry determines a class in $\pi_{2c}\TMF$. Its image in $\MF$ is $Z_{RR}(V) \eta^{2c}$, where $Z_{RR}(V) \in \bZ$ denotes the supersymmetric index of $V$.
\end{conjecture}
\begin{conjecture}\label{holomorphic conjecture}
  In particular, if $V$ is a holomorphic $\cN=1$ SCFT of central charge $c=12k$, then $Z_{RR}(V)$ is divisible by $\frac{24}{\gcd(k,24)}$.
\end{conjecture}

One may prove this Conjecture when $k\leq 2$ by inspecting the classification of holomorphic VOAs of small central charge:
\begin{proof}[Proof of Conjecture~\ref{holomorphic conjecture} for $k\leq 2$]
  When $k=1$, the only holomorphic $\cN=1$ VOA of central charge $c=12$ and nonzero index is Duncan's supermoonshine SCFT $V^{f\natural}$ \cite{MR2352133}; see Examples~\ref{eg.scfts with c=12} and~\ref{supermoonshine continued}. Its index is $24$.
    
  When $k=2$, Conjecture~\ref{holomorphic conjecture} holds even when there is no supersymmetry. Indeed,
  suppose $V$ is an actually-fermionic holomorphic VOA of central charge $c=24$. Its \define{bosonic neighbors} $V^\pm$ are the two results of gauging the fermion parity operator $(-1)^f$ on $V$. This symmetry is nonanomalous, and the results of gauging it are bosonic, because the central charge is divisible by~$8$. (This generalizes the ``even neighbors'' of an odd lattice; c.f.\ \cite[Chapter~17]{MR1662447}.)
  Just as in the proof of Theorem~\ref{thm.index of a code}, we find $Z_{RR}(V) = Z(V^+) - Z(V^-)$, where $Z(V^\pm)$ denotes the ordinary partition function of $V^\pm$. Being holmorphic VOAs of central charge $c=24$, these $V^\pm$ are highly constrained \cite{MR1213740}. In particular, their partition functions are of the form $Z(V^\pm) = J + \dim (V^\pm_1)$ where $J(q) = q^{-1} + O(q) = c_4^3\Delta^{-1} - 744$ is the normalized $\mathrm{SL}(2,\bZ)$ hauptmodule and $\dim (V^\pm_1)$ are the dimensions of the horizontal Lie algebras in $V^\pm$. By inspecting Schellekens' list \cite{MR1213740}, we find that $\dim (V^\pm_1)$, hence $Z_{RR}$, is always divisible by $12$.
\end{proof}

The converse of Conjecture~\ref{holomorphic conjecture} is our \mqlink, restated here:

\begin{mainquestion}\label{mainquestion}
For each $k$, does there exist a holomorphic $\cN=1$ VOA with central charge $c=12k$ and  index $\frac{24}{\gcd(k,24)}$?
\end{mainquestion}

Our ternary code methods answer the \mqlink\ in the affirmative for $k\leq 5$:

\begin{theorem}\label{thm.answers}
  Let $C$ be a self-dual ternary code of length $c=12k$ and index $24$ which admits an automorphism $g\in \Aut(C) \subset 2^c{:}S_c$ of order $m=\gcd(24,k)$. Suppose that:
  \begin{enumerate}
    \item If $m$ is even, then $g^{m/2} = -1 \in \Aut(C)$ is the central element. \label{even condition}
    \item If $m$ is divisible by $3$, then $g^{m/3}$ fixes the same number of even and odd maximal codewords, and the image of $g^{m/3}$ in $S_c$ has cycle shape $3^{c/3}$. \label{odd condition}
  \end{enumerate}
  Lift $g$ to an automorphism of $V_{\Lambda(C)}$ via Theorem~\ref{thm.symmetries lift}.
  Then $V_{\Lambda(C)} \sslash  \langle g \rangle$ answers the \mqlink: its index is $\frac{24}m$. 
  In particular, the codes listed in Example~\ref{eg.somecodes} provide answers to the \mqlink\ for $k\leq5$.
\end{theorem}

\begin{proof}
  We first check that the action of $\langle g \rangle$ on $V_{\Lambda(C)}$ is nonanomalous. We will do so by following Example~\ref{eg.anomaly of a lattice automorphism}.
  
  When $m$ is even, we must first check that the anomaly has no Majorana layer, i.e.\ we must check that $\langle \chi, (1-g)\chi\rangle \in 4\bZ$, where $\chi = \sqrt3(1,1,\dots,1)$.
  Condition (\ref{even condition}) says that the $(m/2)$th power of $g$ is $-1 \in \Aut(C) \subset 2^c{:}S_c$. 
  Together with condition (\ref{odd condition}) when $m = 2^a\times 3$, we find that the image of $g$ in $S_c$ consists entirely of blocks of size $m/2$, and each block changes an odd number of signs. There are $c/(m/2) = 2c/m = 24k/m$ blocks, each contributing $3(m/2) - 2\times\text{odd}$ to $\langle \chi, (1-g)\chi\rangle$. The claim follows.
  
  Let us continue the case when $m$ is even. Then $g$ acts without fixed points on $\Lambda(C)$, and so we may apply Example~\ref{eg.anomaly of a lattice automorphism}. The previous paragraph identifies the Frame shape of $g$: it is $(m/2)^{-24k/m} m^{24k/m}$. Noting that $m$ divides $k$, we see that the 't Hooft anomaly of the action of $g$ is $24$ times that of an automorphism of Frame shape $(m/2)^{-k/m} m^{k/m}$, and so vanishes.
  
  The final case to check is when $m=3$. Then $g$ preserves a sublattice, and so there are multiple lifts of $g \in \Aut(\Lambda(C))$ to an automorphism of $V_{\Lambda(C)}$. The lift we are using is \define{standard} in the sense that it commutes with a choice of lift of $-1 \in \Aut(\Lambda(C))$, and so the Frame shape formula from Example~\ref{eg.anomaly of a lattice automorphism} applies \cite[Chapter 5]{MollerThesis}.
  The Frame shape of $g$ is $3^{12k/m}$, and $k$ is divisible by $m$, completing the verification that the anomaly vanishes.
  
It remains to actually calculate $Z_{RR}(V_{\Lambda(C)} \sslash  \langle g\rangle)$.

Let $V$ be any holomorphic VOA and $G \subset \Aut(V)$ a nonanomalous finite subgroup, with chosen trivialization of the anomaly, so that the  orbifold $V \sslash  G$ is defined. 
   The Ramond--Ramond partition function $V \sslash  G$ can be computed as a sum over ``twisted-twined'' partition functions of the action of $G$ on $V$:
  $$ Z_{RR}( V \sslash  G ) = \frac1{|G|} \; \sum_{\substack{g',g'' \in G \\[2pt] [g',g''] = 1}} Z_{RR} \left(\!\!
  \begin{tikzpicture}[baseline=(b)]
    \path (0,-3pt) coordinate (b); \draw[-] (0,0) -- node[auto] {$\scriptstyle g''$} (0,10pt) -- (15pt,10pt) -- (15pt,0) -- coordinate (s)  (0,0); \path (s) ++(0,-8pt)  node[anchor=base] {$\scriptstyle g'$};
  \end{tikzpicture}
 \right) $$
 The notation means the following. One takes the ``$g'$-twisted Ramond-sector'' $V$-module, and computes the super (i.e.\ signed) trace of the action of $g''$ thereon.  The choice of trivialization of the  anomaly is hidden in the notation: it is used to define the action of $g''$ on the $g'$-twisted Ramond sector (a priori the action of $g''$ is defined only projectively) and to assign fermion parities to the twisted sectors.
 
 We now take $V = V_{\Lambda(C)}$ and $G = \langle g \rangle$ of order $|G| = n$. The ``untwisted'' contribution to $Z_{RR}(V\sslash G)$ is
 $$
 \frac1{m} Z_{RR} \left(\!\!
  \begin{tikzpicture}[baseline=(b)]
    \path (0,-3pt) coordinate (b); \draw[-] (0,0) -- node[auto] {$\scriptstyle 1$} (0,10pt) -- (15pt,10pt) -- (15pt,0) -- coordinate (s)  (0,0); \path (s) ++(0,-8pt)  node[anchor=base] {$\scriptstyle 1$};
  \end{tikzpicture}
 \right)
 = \frac1m Z_{RR}(V) = \frac{24}m.$$
 Thus it suffices to show that the ``twisted'' terms all vanish.
 
 Suppose that $g',g''\in\langle g\rangle$ are not both the identity. The subgroup of $\langle g\rangle$ generated by $g',g''$ is a cyclic group; let $h$ denote a choice of generator. Then $Z_{RR} \left(\!\!
  \begin{tikzpicture}[baseline=(b)]
    \path (0,-3pt) coordinate (b); \draw[-] (0,0) -- node[auto] {$\scriptstyle g''$} (0,10pt) -- (15pt,10pt) -- (15pt,0) -- coordinate (s)  (0,0); \path (s) ++(0,-8pt)  node[anchor=base] {$\scriptstyle g'$};
  \end{tikzpicture}
 \right)$ is related by a modular transformation to $Z_{RR} \left(\!\!
  \begin{tikzpicture}[baseline=(b)]
    \path (0,-3pt) coordinate (b); \draw[-] (0,0) -- node[auto] {$\scriptstyle h$} (0,10pt) -- (15pt,10pt) -- (15pt,0) -- coordinate (s)  (0,0); \path (s) ++(0,-8pt)  node[anchor=base] {$\scriptstyle 1$};
  \end{tikzpicture}
 \right) = \operatorname{tr}_{V_R}\left((-1)^f q^{L_0 - c/24} h\right)$.  So it suffiest to show that $Z_{RR} \left(\!\!
  \begin{tikzpicture}[baseline=(b)]
    \path (0,-3pt) coordinate (b); \draw[-] (0,0) -- node[auto] {$\scriptstyle h$} (0,10pt) -- (15pt,10pt) -- (15pt,0) -- coordinate (s)  (0,0); \path (s) ++(0,-8pt)  node[anchor=base] {$\scriptstyle 1$};
  \end{tikzpicture}
 \right) = 0$ for $h \neq 1$.
 
 By construction, the action of $h$ commutes with the supersymmetry, and so, just as for $Z_{RR}(V)$, the contributions to $Z_{RR} \left(\!\!
  \begin{tikzpicture}[baseline=(b)]
    \path (0,-3pt) coordinate (b); \draw[-] (0,0) -- node[auto] {$\scriptstyle h$} (0,10pt) -- (15pt,10pt) -- (15pt,0) -- coordinate (s)  (0,0); \path (s) ++(0,-8pt)  node[anchor=base] {$\scriptstyle 1$};
  \end{tikzpicture}
 \right)$ from non-ground states cancel: $Z_{RR} \left(\!\!
  \begin{tikzpicture}[baseline=(b)]
    \path (0,-3pt) coordinate (b); \draw[-] (0,0) -- node[auto] {$\scriptstyle h$} (0,10pt) -- (15pt,10pt) -- (15pt,0) -- coordinate (s)  (0,0); \path (s) ++(0,-8pt)  node[anchor=base] {$\scriptstyle 1$};
  \end{tikzpicture}
 \right)$
 is simply the trace of $(-1)^fh$ acting on the Ramond-sector ground states.
 In the proof of Proposition~\ref{prop.index=index} we identified the Ramond-sector ground states with maximal code words in $C$. Thus the action of $h$ on the Ramond-sector ground states lifts the permutation action of $h$ on the maximal codewords. In particular, $\tr_{\text{ground states}}\left((-1)^fh\right)$ receives contributions only from the fixed points of the permutation action of $h$ on $\{\text{maximal codewords}\}$.
 
 If $h$ has even order, then, by the first condition, some power of $h$ acts by the central element $-1 \in \Aut(C)$, and so $h$ has no fixed points. If $h$ has odd order, 
  then there is no sign ambiguity when lifting the action of $h$ from $\Aut(C)$ to $\Aut(V_\Lambda)$, and $\tr_{\text{ground states}}\left((-1)^fh\right)$ is simply a signed count of $h$-fixed maximal codewords. The second condition assures that this signed count vanishes.
\end{proof}

\begin{examplenodiamond} \label{eg.somecodes}
We end by listing codes satisfying the conditions of Theorem~\ref{thm.answers} for $k\leq 5$. They were found by randomly generating self-dual codes with an appropriate automorphism and then calculating the index, repeating the search until one with index $24$ turned up. We received significant help from Noam D.\ Elkies. In particular, he explained to us how to run such a search in a reasonable amount of time, and provided the $k=5$ solution.
\begin{enumerate}
  \item The only, up to signed permutations of the coordinates, self-dual ternary code of length $12$ and non-zero index is the (extended) Ternary Golay code from \cite{Golay}. Its index is~$24$. It can be presented by the generator matrix:
  $$ 
  \left( \begin{array}{cccccccccccc}
1 & . & . & . & . & . & . & 1 & 1 & 1 & 2 & 1\\
. & 1 & . & . & . & . & 2 & 1 & . & 2 & 1 & 1\\
. & . & 1 & . & . & . & 1 & 2 & 1 & 2 & . & 1\\
. & . & . & 1 & . & . & 1 & . & 2 & 1 & 1 & 1\\
. & . & . & . & 1 & . & 1 & 1 & 2 & 2 & 2 & .\\
. & . & . & . & . & 1 & 1 & 1 & 1 & . & 1 & 2\\
  \end{array} \right) $$
  \item  The following generator matrix spans a code of length $24$ and index $24$:  \footnotesize
  $$ \left( \begin{array}{cccccccccccccccccccccccc}
1 & . & . & . & . & . & . & . & . & . & . & . & 1 & 2 & . & . & 2 & 1 & 1 & 1 & 2 & 1 & . & .\\
. & 1 & . & . & . & . & . & . & . & . & . & . & . & . & . & 1 & 1 & 1 & 1 & 1 & . & 1 & 2 & 1\\
. & . & 1 & . & . & . & . & . & . & . & . & . & 1 & 2 & . & 1 & 1 & . & 2 & . & 1 & 1 & 2 & .\\
. & . & . & 1 & . & . & . & . & . & . & . & . & 1 & 1 & 2 & . & 1 & 2 & . & 2 & 1 & 1 & . & .\\
. & . & . & . & 1 & . & . & . & . & . & . & . & . & 2 & 1 & 1 & . & . & 2 & 2 & . & 1 & 1 & 1\\
. & . & . & . & . & 1 & . & . & . & . & . & . & 1 & 1 & 1 & 2 & 2 & 1 & 1 & . & 2 & 2 & 1 & 2\\
. & . & . & . & . & . & 1 & . & . & . & . & . & . & 2 & 2 & 1 & 1 & . & . & . & 1 & 1 & 2 & 2\\
. & . & . & . & . & . & . & 1 & . & . & . & . & 2 & . & 1 & 2 & 2 & . & 2 & 2 & . & 2 & . & 2\\
. & . & . & . & . & . & . & . & 1 & . & . & . & 2 & 1 & 1 & 1 & 1 & 2 & 1 & . & 1 & 1 & 1 & 1\\
. & . & . & . & . & . & . & . & . & 1 & . & . & 2 & . & . & . & . & 1 & 1 & 2 & 2 & 2 & 2 & 2\\
. & . & . & . & . & . & . & . & . & . & 1 & . & 2 & 2 & 1 & . & . & 1 & . & 2 & . & . & . & .\\
. & . & . & . & . & . & . & . & . & . & . & 1 & . & . & 1 & 1 & . & 1 & . & 2 & 2 & 2 & 2 & 2
  \end{array} \right) $$
  \normalsize
  We set $g = -1 \in \Aut(C)$.
  \item 
  The following generator matrix spans a self-dual ternary code of length $36$ and index $24$. \tiny
  $$ \hspace{-1in} \left( \begin{array}{cccccccccccccccccccccccccccccccccccc}
1\!\! & .\!\! & .\!\! & .\!\! & .\!\! & .\!\! & .\!\! & .\!\! & .\!\! & .\!\! & .\!\! & .\!\! & .\!\! & .\!\! & .\!\! & .\!\! & .\!\! & 2\!\! & 2\!\! & .\!\! & 2\!\! & 2\!\! & 2\!\! & 1\!\! & .\!\! & .\!\! & 2\!\! & .\!\! & 1\!\! & 1\!\! & 1\!\! & 1\!\! & 1\!\! & 2\!\! & .\!\! & 2\\
.\!\! & 1\!\! & .\!\! & .\!\! & .\!\! & .\!\! & .\!\! & .\!\! & .\!\! & .\!\! & .\!\! & .\!\! & .\!\! & .\!\! & .\!\! & .\!\! & .\!\! & 1\!\! & 2\!\! & .\!\! & .\!\! & 2\!\! & .\!\! & .\!\! & 2\!\! & .\!\! & .\!\! & .\!\! & 2\!\! & 1\!\! & .\!\! & .\!\! & .\!\! & 1\!\! & 1\!\! & .\\
.\!\! & .\!\! & 1\!\! & .\!\! & .\!\! & .\!\! & .\!\! & .\!\! & .\!\! & .\!\! & .\!\! & .\!\! & .\!\! & .\!\! & .\!\! & .\!\! & .\!\! & 1\!\! & .\!\! & .\!\! & .\!\! & 2\!\! & 1\!\! & .\!\! & 1\!\! & 2\!\! & 2\!\! & .\!\! & 1\!\! & .\!\! & 2\!\! & 1\!\! & .\!\! & 1\!\! & 2\!\! & .\\
.\!\! & .\!\! & .\!\! & 1\!\! & .\!\! & .\!\! & .\!\! & .\!\! & .\!\! & .\!\! & .\!\! & .\!\! & .\!\! & .\!\! & .\!\! & .\!\! & .\!\! & 1\!\! & .\!\! & .\!\! & .\!\! & .\!\! & 1\!\! & .\!\! & 2\!\! & 1\!\! & 1\!\! & 1\!\! & .\!\! & 1\!\! & 2\!\! & 2\!\! & .\!\! & 2\!\! & 1\!\! & .\\
.\!\! & .\!\! & .\!\! & .\!\! & 1\!\! & .\!\! & .\!\! & .\!\! & .\!\! & .\!\! & .\!\! & .\!\! & .\!\! & .\!\! & .\!\! & .\!\! & .\!\! & .\!\! & 2\!\! & .\!\! & 1\!\! & 2\!\! & 2\!\! & 1\!\! & .\!\! & 2\!\! & 2\!\! & 2\!\! & .\!\! & 2\!\! & .\!\! & 1\!\! & 1\!\! & 1\!\! & 1\!\! & 1\\
.\!\! & .\!\! & .\!\! & .\!\! & .\!\! & 1\!\! & .\!\! & .\!\! & .\!\! & .\!\! & .\!\! & .\!\! & .\!\! & .\!\! & .\!\! & .\!\! & .\!\! & .\!\! & 2\!\! & .\!\! & 2\!\! & .\!\! & 2\!\! & 2\!\! & 2\!\! & .\!\! & .\!\! & 2\!\! & .\!\! & 2\!\! & 1\!\! & 1\!\! & .\!\! & 2\!\! & .\!\! & 1\\
.\!\! & .\!\! & .\!\! & .\!\! & .\!\! & .\!\! & 1\!\! & .\!\! & .\!\! & .\!\! & .\!\! & .\!\! & .\!\! & .\!\! & .\!\! & .\!\! & .\!\! & 1\!\! & .\!\! & .\!\! & 1\!\! & .\!\! & 1\!\! & 1\!\! & .\!\! & 2\!\! & 1\!\! & 1\!\! & 2\!\! & 1\!\! & 1\!\! & .\!\! & .\!\! & .\!\! & .\!\! & 1\\
.\!\! & .\!\! & .\!\! & .\!\! & .\!\! & .\!\! & .\!\! & 1\!\! & .\!\! & .\!\! & .\!\! & .\!\! & .\!\! & .\!\! & .\!\! & .\!\! & .\!\! & 1\!\! & 2\!\! & .\!\! & .\!\! & .\!\! & 2\!\! & 2\!\! & 2\!\! & .\!\! & 1\!\! & 2\!\! & .\!\! & .\!\! & 1\!\! & 2\!\! & .\!\! & 2\!\! & 1\!\! & .\\
.\!\! & .\!\! & .\!\! & .\!\! & .\!\! & .\!\! & .\!\! & .\!\! & 1\!\! & .\!\! & .\!\! & .\!\! & .\!\! & .\!\! & .\!\! & .\!\! & .\!\! & 2\!\! & 1\!\! & .\!\! & 2\!\! & .\!\! & .\!\! & 1\!\! & .\!\! & 1\!\! & 1\!\! & 2\!\! & 1\!\! & 1\!\! & 1\!\! & 2\!\! & 1\!\! & .\!\! & 2\!\! & 1\\
.\!\! & .\!\! & .\!\! & .\!\! & .\!\! & .\!\! & .\!\! & .\!\! & .\!\! & 1\!\! & .\!\! & .\!\! & .\!\! & .\!\! & .\!\! & .\!\! & .\!\! & 2\!\! & 1\!\! & .\!\! & 1\!\! & .\!\! & 1\!\! & 1\!\! & 1\!\! & 1\!\! & .\!\! & 1\!\! & .\!\! & .\!\! & 2\!\! & 1\!\! & 2\!\! & 2\!\! & 1\!\! & 2\\
.\!\! & .\!\! & .\!\! & .\!\! & .\!\! & .\!\! & .\!\! & .\!\! & .\!\! & .\!\! & 1\!\! & .\!\! & .\!\! & .\!\! & .\!\! & .\!\! & .\!\! & 1\!\! & 2\!\! & .\!\! & .\!\! & 2\!\! & .\!\! & 1\!\! & 1\!\! & .\!\! & 2\!\! & 2\!\! & 2\!\! & 1\!\! & 1\!\! & 1\!\! & 2\!\! & 1\!\! & 1\!\! & .\\
.\!\! & .\!\! & .\!\! & .\!\! & .\!\! & .\!\! & .\!\! & .\!\! & .\!\! & .\!\! & .\!\! & 1\!\! & .\!\! & .\!\! & .\!\! & .\!\! & .\!\! & 1\!\! & .\!\! & .\!\! & 1\!\! & 2\!\! & .\!\! & 1\!\! & .\!\! & .\!\! & 1\!\! & 1\!\! & 2\!\! & .\!\! & .\!\! & .\!\! & 2\!\! & .\!\! & .\!\! & .\\
.\!\! & .\!\! & .\!\! & .\!\! & .\!\! & .\!\! & .\!\! & .\!\! & .\!\! & .\!\! & .\!\! & .\!\! & 1\!\! & .\!\! & .\!\! & .\!\! & .\!\! & .\!\! & .\!\! & .\!\! & 2\!\! & 2\!\! & 1\!\! & 1\!\! & 1\!\! & 1\!\! & 1\!\! & 2\!\! & .\!\! & 2\!\! & 2\!\! & 2\!\! & 1\!\! & .\!\! & 1\!\! & 2\\
.\!\! & .\!\! & .\!\! & .\!\! & .\!\! & .\!\! & .\!\! & .\!\! & .\!\! & .\!\! & .\!\! & .\!\! & .\!\! & 1\!\! & .\!\! & .\!\! & .\!\! & 2\!\! & 2\!\! & .\!\! & .\!\! & 2\!\! & 1\!\! & 1\!\! & 2\!\! & 1\!\! & 2\!\! & .\!\! & 2\!\! & 1\!\! & 2\!\! & 2\!\! & 1\!\! & .\!\! & 2\!\! & .\\
.\!\! & .\!\! & .\!\! & .\!\! & .\!\! & .\!\! & .\!\! & .\!\! & .\!\! & .\!\! & .\!\! & .\!\! & .\!\! & .\!\! & 1\!\! & .\!\! & .\!\! & 2\!\! & .\!\! & .\!\! & .\!\! & 2\!\! & 1\!\! & .\!\! & 1\!\! & 1\!\! & .\!\! & 1\!\! & .\!\! & 1\!\! & .\!\! & .\!\! & 1\!\! & 1\!\! & 2\!\! & 2\\
.\!\! & .\!\! & .\!\! & .\!\! & .\!\! & .\!\! & .\!\! & .\!\! & .\!\! & .\!\! & .\!\! & .\!\! & .\!\! & .\!\! & .\!\! & 1\!\! & .\!\! & 2\!\! & 2\!\! & .\!\! & 2\!\! & 2\!\! & 2\!\! & 2\!\! & 1\!\! & .\!\! & 2\!\! & 2\!\! & 1\!\! & 1\!\! & 2\!\! & 2\!\! & 1\!\! & 2\!\! & 1\!\! & 1\\
.\!\! & .\!\! & .\!\! & .\!\! & .\!\! & .\!\! & .\!\! & .\!\! & .\!\! & .\!\! & .\!\! & .\!\! & .\!\! & .\!\! & .\!\! & .\!\! & 1\!\! & 2\!\! & 1\!\! & .\!\! & 1\!\! & .\!\! & 1\!\! & 2\!\! & .\!\! & 1\!\! & 1\!\! & 1\!\! & 1\!\! & 2\!\! & 2\!\! & 1\!\! & .\!\! & .\!\! & 1\!\! & 2\\
.\!\! & .\!\! & .\!\! & .\!\! & .\!\! & .\!\! & .\!\! & .\!\! & .\!\! & .\!\! & .\!\! & .\!\! & .\!\! & .\!\! & .\!\! & .\!\! & .\!\! & .\!\! & .\!\! & 1\!\! & 2\!\! & 1\!\! & 1\!\! & 2\!\! & .\!\! & 1\!\! & 2\!\! & 1\!\! & 2\!\! & 1\!\! & 1\!\! & 1\!\! & 1\!\! & .\!\! & 2\!\! & 2
  \end{array} \right) \hspace{-1in} $$
  \normalsize
  This code is invariant under the permutation $g : e_i \mapsto e_{i + 12}$ of the coordinates, where the coordinates are $e_1,\dots,e_{36}$ and the sum is considered mod $36$. The fixed subcode, when considered as a code inside $(\bF_3^{36})^{\langle g \rangle} = \bF_3^{12}$, has the following generator matrix:
  $$ \left(\begin{array}{cccccccccccc}
  . & . & . & . & . & 2 & 2 & . & 1 & 1 & 1 & 1\\
. & 1 & . & . & . & 2 & 1 & . & 2 & 2 & 2 & .\\
. & . & 1 & . & . & . & . & . & 2 & 1 & . & .\\
1 & . & . & 1 & . & . & 1 & . & 1 & 1 & 2 & .\\
1 & . & . & . & 1 & . & 2 & . & 2 & 2 & 1 & .\\
. & . & . & . & . & . & 1 & 1 & . & . & 1 & .
  \end{array}\right)$$
  This is a copy of the self-dual ternary code of length $12$ with eight even and eight odd maximal codewords, and so this pair $(C,g)$ satisfies condition (2) of Theorem~\ref{thm.answers}.
  \item
  The following generator matrix spans a self-dual ternary code of length $48$ and index $24$. \tiny
  $$ \hspace{-1in} \left(\begin{array}{cccccccccccccccccccccccccccccccccccccccccccccccc}
  1\!\!\!\! & .\!\!\!\! & .\!\!\!\! & .\!\!\!\! & .\!\!\!\! & .\!\!\!\! & .\!\!\!\! & .\!\!\!\! & .\!\!\!\! & .\!\!\!\! & .\!\!\!\! & .\!\!\!\! & .\!\!\!\! & .\!\!\!\! & .\!\!\!\! & .\!\!\!\! & .\!\!\!\! & .\!\!\!\! & .\!\!\!\! & .\!\!\!\! & .\!\!\!\! & .\!\!\!\! & .\!\!\!\! & .\!\!\!\! & 1\!\!\!\! & .\!\!\!\! & .\!\!\!\! & 1\!\!\!\! & 2\!\!\!\! & 1\!\!\!\! & .\!\!\!\! & 2\!\!\!\! & 1\!\!\!\! & 2\!\!\!\! & 1\!\!\!\! & 2\!\!\!\! & 2\!\!\!\! & .\!\!\!\! & .\!\!\!\! & 1\!\!\!\! & 2\!\!\!\! & 2\!\!\!\! & 2\!\!\!\! & .\!\!\!\! & 2\!\!\!\! & .\!\!\!\! & 2\!\!\!\! & 2\\
.\!\!\!\! & 1\!\!\!\! & .\!\!\!\! & .\!\!\!\! & .\!\!\!\! & .\!\!\!\! & .\!\!\!\! & .\!\!\!\! & .\!\!\!\! & .\!\!\!\! & .\!\!\!\! & .\!\!\!\! & .\!\!\!\! & .\!\!\!\! & .\!\!\!\! & .\!\!\!\! & .\!\!\!\! & .\!\!\!\! & .\!\!\!\! & .\!\!\!\! & .\!\!\!\! & .\!\!\!\! & .\!\!\!\! & .\!\!\!\! & .\!\!\!\! & .\!\!\!\! & .\!\!\!\! & 1\!\!\!\! & 2\!\!\!\! & 2\!\!\!\! & 2\!\!\!\! & 2\!\!\!\! & 2\!\!\!\! & .\!\!\!\! & 1\!\!\!\! & .\!\!\!\! & .\!\!\!\! & 2\!\!\!\! & .\!\!\!\! & 1\!\!\!\! & 2\!\!\!\! & 1\!\!\!\! & .\!\!\!\! & 2\!\!\!\! & 2\!\!\!\! & .\!\!\!\! & .\!\!\!\! & 1\\
.\!\!\!\! & .\!\!\!\! & 1\!\!\!\! & .\!\!\!\! & .\!\!\!\! & .\!\!\!\! & .\!\!\!\! & .\!\!\!\! & .\!\!\!\! & .\!\!\!\! & .\!\!\!\! & .\!\!\!\! & .\!\!\!\! & .\!\!\!\! & .\!\!\!\! & .\!\!\!\! & .\!\!\!\! & .\!\!\!\! & .\!\!\!\! & .\!\!\!\! & .\!\!\!\! & .\!\!\!\! & .\!\!\!\! & .\!\!\!\! & .\!\!\!\! & .\!\!\!\! & 2\!\!\!\! & 1\!\!\!\! & 2\!\!\!\! & .\!\!\!\! & .\!\!\!\! & 2\!\!\!\! & .\!\!\!\! & .\!\!\!\! & 2\!\!\!\! & 2\!\!\!\! & 2\!\!\!\! & 1\!\!\!\! & 2\!\!\!\! & 2\!\!\!\! & .\!\!\!\! & .\!\!\!\! & .\!\!\!\! & .\!\!\!\! & 1\!\!\!\! & 2\!\!\!\! & 1\!\!\!\! & 1\\
.\!\!\!\! & .\!\!\!\! & .\!\!\!\! & 1\!\!\!\! & .\!\!\!\! & .\!\!\!\! & .\!\!\!\! & .\!\!\!\! & .\!\!\!\! & .\!\!\!\! & .\!\!\!\! & .\!\!\!\! & .\!\!\!\! & .\!\!\!\! & .\!\!\!\! & .\!\!\!\! & .\!\!\!\! & .\!\!\!\! & .\!\!\!\! & .\!\!\!\! & .\!\!\!\! & .\!\!\!\! & .\!\!\!\! & .\!\!\!\! & 1\!\!\!\! & 1\!\!\!\! & 1\!\!\!\! & .\!\!\!\! & 1\!\!\!\! & 1\!\!\!\! & .\!\!\!\! & .\!\!\!\! & 1\!\!\!\! & .\!\!\!\! & 2\!\!\!\! & .\!\!\!\! & 1\!\!\!\! & 1\!\!\!\! & .\!\!\!\! & 1\!\!\!\! & 1\!\!\!\! & 2\!\!\!\! & 1\!\!\!\! & 1\!\!\!\! & 1\!\!\!\! & 1\!\!\!\! & .\!\!\!\! & 2\\
.\!\!\!\! & .\!\!\!\! & .\!\!\!\! & .\!\!\!\! & 1\!\!\!\! & .\!\!\!\! & .\!\!\!\! & .\!\!\!\! & .\!\!\!\! & .\!\!\!\! & .\!\!\!\! & .\!\!\!\! & .\!\!\!\! & .\!\!\!\! & .\!\!\!\! & .\!\!\!\! & .\!\!\!\! & .\!\!\!\! & .\!\!\!\! & .\!\!\!\! & .\!\!\!\! & .\!\!\!\! & .\!\!\!\! & .\!\!\!\! & 2\!\!\!\! & 2\!\!\!\! & 2\!\!\!\! & 1\!\!\!\! & .\!\!\!\! & 2\!\!\!\! & .\!\!\!\! & 2\!\!\!\! & 1\!\!\!\! & 2\!\!\!\! & 1\!\!\!\! & 1\!\!\!\! & 1\!\!\!\! & 1\!\!\!\! & 1\!\!\!\! & 2\!\!\!\! & 1\!\!\!\! & 1\!\!\!\! & 1\!\!\!\! & .\!\!\!\! & .\!\!\!\! & 1\!\!\!\! & 1\!\!\!\! & 2\\
.\!\!\!\! & .\!\!\!\! & .\!\!\!\! & .\!\!\!\! & .\!\!\!\! & 1\!\!\!\! & .\!\!\!\! & .\!\!\!\! & .\!\!\!\! & .\!\!\!\! & .\!\!\!\! & .\!\!\!\! & .\!\!\!\! & .\!\!\!\! & .\!\!\!\! & .\!\!\!\! & .\!\!\!\! & .\!\!\!\! & .\!\!\!\! & .\!\!\!\! & .\!\!\!\! & .\!\!\!\! & .\!\!\!\! & .\!\!\!\! & 1\!\!\!\! & 2\!\!\!\! & .\!\!\!\! & 1\!\!\!\! & 2\!\!\!\! & 1\!\!\!\! & 2\!\!\!\! & .\!\!\!\! & .\!\!\!\! & .\!\!\!\! & 2\!\!\!\! & 1\!\!\!\! & 1\!\!\!\! & 1\!\!\!\! & 1\!\!\!\! & 2\!\!\!\! & 2\!\!\!\! & 1\!\!\!\! & 2\!\!\!\! & .\!\!\!\! & .\!\!\!\! & .\!\!\!\! & 2\!\!\!\! & 2\\
.\!\!\!\! & .\!\!\!\! & .\!\!\!\! & .\!\!\!\! & .\!\!\!\! & .\!\!\!\! & 1\!\!\!\! & .\!\!\!\! & .\!\!\!\! & .\!\!\!\! & .\!\!\!\! & .\!\!\!\! & .\!\!\!\! & .\!\!\!\! & .\!\!\!\! & .\!\!\!\! & .\!\!\!\! & .\!\!\!\! & .\!\!\!\! & .\!\!\!\! & .\!\!\!\! & .\!\!\!\! & .\!\!\!\! & .\!\!\!\! & .\!\!\!\! & 2\!\!\!\! & .\!\!\!\! & .\!\!\!\! & .\!\!\!\! & 2\!\!\!\! & 1\!\!\!\! & 2\!\!\!\! & 1\!\!\!\! & 1\!\!\!\! & 2\!\!\!\! & .\!\!\!\! & .\!\!\!\! & 2\!\!\!\! & 1\!\!\!\! & 2\!\!\!\! & .\!\!\!\! & 2\!\!\!\! & .\!\!\!\! & .\!\!\!\! & 2\!\!\!\! & .\!\!\!\! & 2\!\!\!\! & 1\\
.\!\!\!\! & .\!\!\!\! & .\!\!\!\! & .\!\!\!\! & .\!\!\!\! & .\!\!\!\! & .\!\!\!\! & 1\!\!\!\! & .\!\!\!\! & .\!\!\!\! & .\!\!\!\! & .\!\!\!\! & .\!\!\!\! & .\!\!\!\! & .\!\!\!\! & .\!\!\!\! & .\!\!\!\! & .\!\!\!\! & .\!\!\!\! & .\!\!\!\! & .\!\!\!\! & .\!\!\!\! & .\!\!\!\! & .\!\!\!\! & 2\!\!\!\! & 2\!\!\!\! & 2\!\!\!\! & .\!\!\!\! & 2\!\!\!\! & .\!\!\!\! & 2\!\!\!\! & 2\!\!\!\! & 2\!\!\!\! & 2\!\!\!\! & 1\!\!\!\! & 2\!\!\!\! & 1\!\!\!\! & .\!\!\!\! & 1\!\!\!\! & 1\!\!\!\! & .\!\!\!\! & 2\!\!\!\! & 1\!\!\!\! & .\!\!\!\! & 1\!\!\!\! & .\!\!\!\! & .\!\!\!\! & 2\\
.\!\!\!\! & .\!\!\!\! & .\!\!\!\! & .\!\!\!\! & .\!\!\!\! & .\!\!\!\! & .\!\!\!\! & .\!\!\!\! & 1\!\!\!\! & .\!\!\!\! & .\!\!\!\! & .\!\!\!\! & .\!\!\!\! & .\!\!\!\! & .\!\!\!\! & .\!\!\!\! & .\!\!\!\! & .\!\!\!\! & .\!\!\!\! & .\!\!\!\! & .\!\!\!\! & .\!\!\!\! & .\!\!\!\! & .\!\!\!\! & 1\!\!\!\! & 2\!\!\!\! & .\!\!\!\! & 1\!\!\!\! & 1\!\!\!\! & .\!\!\!\! & 1\!\!\!\! & 2\!\!\!\! & .\!\!\!\! & 1\!\!\!\! & 1\!\!\!\! & .\!\!\!\! & .\!\!\!\! & 1\!\!\!\! & 2\!\!\!\! & .\!\!\!\! & 1\!\!\!\! & 1\!\!\!\! & 1\!\!\!\! & 1\!\!\!\! & .\!\!\!\! & 1\!\!\!\! & 1\!\!\!\! & 1\\
.\!\!\!\! & .\!\!\!\! & .\!\!\!\! & .\!\!\!\! & .\!\!\!\! & .\!\!\!\! & .\!\!\!\! & .\!\!\!\! & .\!\!\!\! & 1\!\!\!\! & .\!\!\!\! & .\!\!\!\! & .\!\!\!\! & .\!\!\!\! & .\!\!\!\! & .\!\!\!\! & .\!\!\!\! & .\!\!\!\! & .\!\!\!\! & .\!\!\!\! & .\!\!\!\! & .\!\!\!\! & .\!\!\!\! & .\!\!\!\! & 2\!\!\!\! & .\!\!\!\! & .\!\!\!\! & .\!\!\!\! & 2\!\!\!\! & .\!\!\!\! & 1\!\!\!\! & 2\!\!\!\! & 1\!\!\!\! & .\!\!\!\! & 1\!\!\!\! & 2\!\!\!\! & .\!\!\!\! & 1\!\!\!\! & 1\!\!\!\! & 2\!\!\!\! & .\!\!\!\! & .\!\!\!\! & 1\!\!\!\! & .\!\!\!\! & .\!\!\!\! & 2\!\!\!\! & 1\!\!\!\! & 1\\
.\!\!\!\! & .\!\!\!\! & .\!\!\!\! & .\!\!\!\! & .\!\!\!\! & .\!\!\!\! & .\!\!\!\! & .\!\!\!\! & .\!\!\!\! & .\!\!\!\! & 1\!\!\!\! & .\!\!\!\! & .\!\!\!\! & .\!\!\!\! & .\!\!\!\! & .\!\!\!\! & .\!\!\!\! & .\!\!\!\! & .\!\!\!\! & .\!\!\!\! & .\!\!\!\! & .\!\!\!\! & .\!\!\!\! & .\!\!\!\! & 1\!\!\!\! & 1\!\!\!\! & 2\!\!\!\! & 2\!\!\!\! & 1\!\!\!\! & 2\!\!\!\! & 2\!\!\!\! & 1\!\!\!\! & 1\!\!\!\! & 1\!\!\!\! & 1\!\!\!\! & 2\!\!\!\! & 2\!\!\!\! & .\!\!\!\! & .\!\!\!\! & .\!\!\!\! & .\!\!\!\! & 1\!\!\!\! & .\!\!\!\! & 2\!\!\!\! & .\!\!\!\! & 1\!\!\!\! & .\!\!\!\! & 2\\
.\!\!\!\! & .\!\!\!\! & .\!\!\!\! & .\!\!\!\! & .\!\!\!\! & .\!\!\!\! & .\!\!\!\! & .\!\!\!\! & .\!\!\!\! & .\!\!\!\! & .\!\!\!\! & 1\!\!\!\! & .\!\!\!\! & .\!\!\!\! & .\!\!\!\! & .\!\!\!\! & .\!\!\!\! & .\!\!\!\! & .\!\!\!\! & .\!\!\!\! & .\!\!\!\! & .\!\!\!\! & .\!\!\!\! & .\!\!\!\! & 2\!\!\!\! & .\!\!\!\! & 2\!\!\!\! & .\!\!\!\! & 1\!\!\!\! & 1\!\!\!\! & .\!\!\!\! & 2\!\!\!\! & .\!\!\!\! & 2\!\!\!\! & 2\!\!\!\! & 2\!\!\!\! & 2\!\!\!\! & .\!\!\!\! & 1\!\!\!\! & .\!\!\!\! & 1\!\!\!\! & 2\!\!\!\! & .\!\!\!\! & .\!\!\!\! & 2\!\!\!\! & 2\!\!\!\! & .\!\!\!\! & .\\
.\!\!\!\! & .\!\!\!\! & .\!\!\!\! & .\!\!\!\! & .\!\!\!\! & .\!\!\!\! & .\!\!\!\! & .\!\!\!\! & .\!\!\!\! & .\!\!\!\! & .\!\!\!\! & .\!\!\!\! & 1\!\!\!\! & .\!\!\!\! & .\!\!\!\! & .\!\!\!\! & .\!\!\!\! & .\!\!\!\! & .\!\!\!\! & .\!\!\!\! & .\!\!\!\! & .\!\!\!\! & .\!\!\!\! & .\!\!\!\! & 2\!\!\!\! & .\!\!\!\! & 2\!\!\!\! & 1\!\!\!\! & 1\!\!\!\! & 1\!\!\!\! & .\!\!\!\! & 1\!\!\!\! & .\!\!\!\! & .\!\!\!\! & 2\!\!\!\! & 2\!\!\!\! & 2\!\!\!\! & 1\!\!\!\! & .\!\!\!\! & 1\!\!\!\! & .\!\!\!\! & .\!\!\!\! & .\!\!\!\! & .\!\!\!\! & 2\!\!\!\! & 1\!\!\!\! & .\!\!\!\! & 2\\
.\!\!\!\! & .\!\!\!\! & .\!\!\!\! & .\!\!\!\! & .\!\!\!\! & .\!\!\!\! & .\!\!\!\! & .\!\!\!\! & .\!\!\!\! & .\!\!\!\! & .\!\!\!\! & .\!\!\!\! & .\!\!\!\! & 1\!\!\!\! & .\!\!\!\! & .\!\!\!\! & .\!\!\!\! & .\!\!\!\! & .\!\!\!\! & .\!\!\!\! & .\!\!\!\! & .\!\!\!\! & .\!\!\!\! & .\!\!\!\! & .\!\!\!\! & 2\!\!\!\! & 1\!\!\!\! & 1\!\!\!\! & 1\!\!\!\! & 1\!\!\!\! & 2\!\!\!\! & .\!\!\!\! & 1\!\!\!\! & 1\!\!\!\! & .\!\!\!\! & .\!\!\!\! & 1\!\!\!\! & 2\!\!\!\! & 1\!\!\!\! & 1\!\!\!\! & 1\!\!\!\! & .\!\!\!\! & 1\!\!\!\! & 2\!\!\!\! & 1\!\!\!\! & .\!\!\!\! & 1\!\!\!\! & .\\
.\!\!\!\! & .\!\!\!\! & .\!\!\!\! & .\!\!\!\! & .\!\!\!\! & .\!\!\!\! & .\!\!\!\! & .\!\!\!\! & .\!\!\!\! & .\!\!\!\! & .\!\!\!\! & .\!\!\!\! & .\!\!\!\! & .\!\!\!\! & 1\!\!\!\! & .\!\!\!\! & .\!\!\!\! & .\!\!\!\! & .\!\!\!\! & .\!\!\!\! & .\!\!\!\! & .\!\!\!\! & .\!\!\!\! & .\!\!\!\! & .\!\!\!\! & .\!\!\!\! & 2\!\!\!\! & .\!\!\!\! & 1\!\!\!\! & 1\!\!\!\! & 1\!\!\!\! & 1\!\!\!\! & 2\!\!\!\! & 1\!\!\!\! & .\!\!\!\! & 1\!\!\!\! & .\!\!\!\! & 1\!\!\!\! & 2\!\!\!\! & .\!\!\!\! & .\!\!\!\! & 2\!\!\!\! & 1\!\!\!\! & 2\!\!\!\! & 1\!\!\!\! & 2\!\!\!\! & 2\!\!\!\! & 2\\
.\!\!\!\! & .\!\!\!\! & .\!\!\!\! & .\!\!\!\! & .\!\!\!\! & .\!\!\!\! & .\!\!\!\! & .\!\!\!\! & .\!\!\!\! & .\!\!\!\! & .\!\!\!\! & .\!\!\!\! & .\!\!\!\! & .\!\!\!\! & .\!\!\!\! & 1\!\!\!\! & .\!\!\!\! & .\!\!\!\! & .\!\!\!\! & .\!\!\!\! & .\!\!\!\! & .\!\!\!\! & .\!\!\!\! & .\!\!\!\! & 1\!\!\!\! & 1\!\!\!\! & 2\!\!\!\! & 1\!\!\!\! & 2\!\!\!\! & 2\!\!\!\! & 2\!\!\!\! & 1\!\!\!\! & .\!\!\!\! & 2\!\!\!\! & .\!\!\!\! & .\!\!\!\! & 1\!\!\!\! & 1\!\!\!\! & .\!\!\!\! & 1\!\!\!\! & .\!\!\!\! & 2\!\!\!\! & 1\!\!\!\! & 1\!\!\!\! & 1\!\!\!\! & .\!\!\!\! & 1\!\!\!\! & .\\
.\!\!\!\! & .\!\!\!\! & .\!\!\!\! & .\!\!\!\! & .\!\!\!\! & .\!\!\!\! & .\!\!\!\! & .\!\!\!\! & .\!\!\!\! & .\!\!\!\! & .\!\!\!\! & .\!\!\!\! & .\!\!\!\! & .\!\!\!\! & .\!\!\!\! & .\!\!\!\! & 1\!\!\!\! & .\!\!\!\! & .\!\!\!\! & .\!\!\!\! & .\!\!\!\! & .\!\!\!\! & .\!\!\!\! & .\!\!\!\! & 2\!\!\!\! & 2\!\!\!\! & .\!\!\!\! & 1\!\!\!\! & 1\!\!\!\! & 2\!\!\!\! & .\!\!\!\! & .\!\!\!\! & 1\!\!\!\! & .\!\!\!\! & .\!\!\!\! & 1\!\!\!\! & .\!\!\!\! & 1\!\!\!\! & .\!\!\!\! & .\!\!\!\! & 1\!\!\!\! & 2\!\!\!\! & 1\!\!\!\! & 1\!\!\!\! & 2\!\!\!\! & .\!\!\!\! & 1\!\!\!\! & .\\
.\!\!\!\! & .\!\!\!\! & .\!\!\!\! & .\!\!\!\! & .\!\!\!\! & .\!\!\!\! & .\!\!\!\! & .\!\!\!\! & .\!\!\!\! & .\!\!\!\! & .\!\!\!\! & .\!\!\!\! & .\!\!\!\! & .\!\!\!\! & .\!\!\!\! & .\!\!\!\! & .\!\!\!\! & 1\!\!\!\! & .\!\!\!\! & .\!\!\!\! & .\!\!\!\! & .\!\!\!\! & .\!\!\!\! & .\!\!\!\! & 2\!\!\!\! & 1\!\!\!\! & .\!\!\!\! & 2\!\!\!\! & 1\!\!\!\! & 1\!\!\!\! & 2\!\!\!\! & 2\!\!\!\! & 1\!\!\!\! & .\!\!\!\! & 1\!\!\!\! & 2\!\!\!\! & .\!\!\!\! & .\!\!\!\! & 2\!\!\!\! & 2\!\!\!\! & 2\!\!\!\! & 2\!\!\!\! & 2\!\!\!\! & 1\!\!\!\! & 1\!\!\!\! & 2\!\!\!\! & 2\!\!\!\! & 1\\
.\!\!\!\! & .\!\!\!\! & .\!\!\!\! & .\!\!\!\! & .\!\!\!\! & .\!\!\!\! & .\!\!\!\! & .\!\!\!\! & .\!\!\!\! & .\!\!\!\! & .\!\!\!\! & .\!\!\!\! & .\!\!\!\! & .\!\!\!\! & .\!\!\!\! & .\!\!\!\! & .\!\!\!\! & .\!\!\!\! & 1\!\!\!\! & .\!\!\!\! & .\!\!\!\! & .\!\!\!\! & .\!\!\!\! & .\!\!\!\! & 2\!\!\!\! & .\!\!\!\! & .\!\!\!\! & 1\!\!\!\! & 1\!\!\!\! & 2\!\!\!\! & .\!\!\!\! & 1\!\!\!\! & 1\!\!\!\! & 1\!\!\!\! & .\!\!\!\! & .\!\!\!\! & .\!\!\!\! & 1\!\!\!\! & 1\!\!\!\! & 1\!\!\!\! & 1\!\!\!\! & 2\!\!\!\! & 1\!\!\!\! & 1\!\!\!\! & 1\!\!\!\! & 1\!\!\!\! & .\!\!\!\! & 2\\
.\!\!\!\! & .\!\!\!\! & .\!\!\!\! & .\!\!\!\! & .\!\!\!\! & .\!\!\!\! & .\!\!\!\! & .\!\!\!\! & .\!\!\!\! & .\!\!\!\! & .\!\!\!\! & .\!\!\!\! & .\!\!\!\! & .\!\!\!\! & .\!\!\!\! & .\!\!\!\! & .\!\!\!\! & .\!\!\!\! & .\!\!\!\! & 1\!\!\!\! & .\!\!\!\! & .\!\!\!\! & .\!\!\!\! & .\!\!\!\! & .\!\!\!\! & 2\!\!\!\! & .\!\!\!\! & 1\!\!\!\! & .\!\!\!\! & .\!\!\!\! & .\!\!\!\! & .\!\!\!\! & 1\!\!\!\! & .\!\!\!\! & 2\!\!\!\! & .\!\!\!\! & .\!\!\!\! & 2\!\!\!\! & 2\!\!\!\! & 1\!\!\!\! & 1\!\!\!\! & 1\!\!\!\! & 1\!\!\!\! & .\!\!\!\! & .\!\!\!\! & .\!\!\!\! & .\!\!\!\! & 2\\
.\!\!\!\! & .\!\!\!\! & .\!\!\!\! & .\!\!\!\! & .\!\!\!\! & .\!\!\!\! & .\!\!\!\! & .\!\!\!\! & .\!\!\!\! & .\!\!\!\! & .\!\!\!\! & .\!\!\!\! & .\!\!\!\! & .\!\!\!\! & .\!\!\!\! & .\!\!\!\! & .\!\!\!\! & .\!\!\!\! & .\!\!\!\! & .\!\!\!\! & 1\!\!\!\! & .\!\!\!\! & .\!\!\!\! & .\!\!\!\! & 2\!\!\!\! & 2\!\!\!\! & 1\!\!\!\! & 1\!\!\!\! & .\!\!\!\! & .\!\!\!\! & 2\!\!\!\! & 1\!\!\!\! & .\!\!\!\! & .\!\!\!\! & .\!\!\!\! & 2\!\!\!\! & 2\!\!\!\! & 1\!\!\!\! & 1\!\!\!\! & 1\!\!\!\! & 2\!\!\!\! & 1\!\!\!\! & 1\!\!\!\! & .\!\!\!\! & 2\!\!\!\! & 1\!\!\!\! & .\!\!\!\! & 2\\
.\!\!\!\! & .\!\!\!\! & .\!\!\!\! & .\!\!\!\! & .\!\!\!\! & .\!\!\!\! & .\!\!\!\! & .\!\!\!\! & .\!\!\!\! & .\!\!\!\! & .\!\!\!\! & .\!\!\!\! & .\!\!\!\! & .\!\!\!\! & .\!\!\!\! & .\!\!\!\! & .\!\!\!\! & .\!\!\!\! & .\!\!\!\! & .\!\!\!\! & .\!\!\!\! & 1\!\!\!\! & .\!\!\!\! & .\!\!\!\! & .\!\!\!\! & .\!\!\!\! & 2\!\!\!\! & 1\!\!\!\! & 1\!\!\!\! & .\!\!\!\! & .\!\!\!\! & .\!\!\!\! & 1\!\!\!\! & 2\!\!\!\! & 1\!\!\!\! & 2\!\!\!\! & 1\!\!\!\! & .\!\!\!\! & 2\!\!\!\! & .\!\!\!\! & .\!\!\!\! & 2\!\!\!\! & 1\!\!\!\! & .\!\!\!\! & 1\!\!\!\! & 1\!\!\!\! & .\!\!\!\! & 2\\
.\!\!\!\! & .\!\!\!\! & .\!\!\!\! & .\!\!\!\! & .\!\!\!\! & .\!\!\!\! & .\!\!\!\! & .\!\!\!\! & .\!\!\!\! & .\!\!\!\! & .\!\!\!\! & .\!\!\!\! & .\!\!\!\! & .\!\!\!\! & .\!\!\!\! & .\!\!\!\! & .\!\!\!\! & .\!\!\!\! & .\!\!\!\! & .\!\!\!\! & .\!\!\!\! & .\!\!\!\! & 1\!\!\!\! & .\!\!\!\! & 2\!\!\!\! & .\!\!\!\! & 1\!\!\!\! & .\!\!\!\! & 1\!\!\!\! & 2\!\!\!\! & 2\!\!\!\! & .\!\!\!\! & 1\!\!\!\! & 1\!\!\!\! & .\!\!\!\! & .\!\!\!\! & .\!\!\!\! & 1\!\!\!\! & 2\!\!\!\! & 1\!\!\!\! & 1\!\!\!\! & 2\!\!\!\! & .\!\!\!\! & .\!\!\!\! & .\!\!\!\! & .\!\!\!\! & 2\!\!\!\! & 2\\
.\!\!\!\! & .\!\!\!\! & .\!\!\!\! & .\!\!\!\! & .\!\!\!\! & .\!\!\!\! & .\!\!\!\! & .\!\!\!\! & .\!\!\!\! & .\!\!\!\! & .\!\!\!\! & .\!\!\!\! & .\!\!\!\! & .\!\!\!\! & .\!\!\!\! & .\!\!\!\! & .\!\!\!\! & .\!\!\!\! & .\!\!\!\! & .\!\!\!\! & .\!\!\!\! & .\!\!\!\! & .\!\!\!\! & 1\!\!\!\! & 2\!\!\!\! & 1\!\!\!\! & 1\!\!\!\! & 2\!\!\!\! & 2\!\!\!\! & 2\!\!\!\! & 1\!\!\!\! & 2\!\!\!\! & 1\!\!\!\! & 1\!\!\!\! & 2\!\!\!\! & .\!\!\!\! & 2\!\!\!\! & .\!\!\!\! & 2\!\!\!\! & .\!\!\!\! & .\!\!\!\! & 1\!\!\!\! & 2\!\!\!\! & 2\!\!\!\! & 2\!\!\!\! & 2\!\!\!\! & 2\!\!\!\! & 1
  \end{array}\right) \hspace{-1in} $$
  \normalsize
  It is invariant under the signed permutation
  $$ e_i \mapsto \begin{cases} e_{i+24}, & i \leq 24, \\
  - e_{i-24}, & i > 24. \end{cases}$$
  \item
The following generator matrix spans a self-dual ternary code of length $60$ and index $24$:
\tiny
  $$ \hspace{-1in} \left(\begin{array}{cccccccccccccccccccccccccccccccccccccccccccccccccccccccccccc}
1\!\!\!\!&.\!\!\!\!&.\!\!\!\!&.\!\!\!\!&.\!\!\!\!&.\!\!\!\!&.\!\!\!\!&.\!\!\!\!&.\!\!\!\!&.\!\!\!\!&.\!\!\!\!&.\!\!\!\!&.\!\!\!\!&.\!\!\!\!&.\!\!\!\!&.\!\!\!\!&.\!\!\!\!&.\!\!\!\!&.\!\!\!\!&.\!\!\!\!&.\!\!\!\!&.\!\!\!\!&.\!\!\!\!&.\!\!\!\!&.\!\!\!\!&.\!\!\!\!&.\!\!\!\!&.\!\!\!\!&.\!\!\!\!&.\!\!\!\!&2\!\!\!\!&1\!\!\!\!&2\!\!\!\!&.\!\!\!\!&1\!\!\!\!&2\!\!\!\!&2\!\!\!\!&1\!\!\!\!&2\!\!\!\!&.\!\!\!\!&2\!\!\!\!&1\!\!\!\!&.\!\!\!\!&2\!\!\!\!&2\!\!\!\!&1\!\!\!\!&2\!\!\!\!&2\!\!\!\!&2\!\!\!\!&2\!\!\!\!&2\!\!\!\!&1\!\!\!\!&1\!\!\!\!&.\!\!\!\!&1\!\!\!\!&1\!\!\!\!&.\!\!\!\!&1\!\!\!\!&.\!\!\!\!&.\\
.\!\!\!\!&1\!\!\!\!&.\!\!\!\!&.\!\!\!\!&.\!\!\!\!&.\!\!\!\!&.\!\!\!\!&.\!\!\!\!&.\!\!\!\!&.\!\!\!\!&.\!\!\!\!&.\!\!\!\!&.\!\!\!\!&.\!\!\!\!&.\!\!\!\!&.\!\!\!\!&.\!\!\!\!&.\!\!\!\!&.\!\!\!\!&.\!\!\!\!&.\!\!\!\!&.\!\!\!\!&.\!\!\!\!&.\!\!\!\!&.\!\!\!\!&.\!\!\!\!&.\!\!\!\!&.\!\!\!\!&.\!\!\!\!&.\!\!\!\!&2\!\!\!\!&2\!\!\!\!&1\!\!\!\!&.\!\!\!\!&.\!\!\!\!&.\!\!\!\!&2\!\!\!\!&.\!\!\!\!&.\!\!\!\!&2\!\!\!\!&1\!\!\!\!&1\!\!\!\!&.\!\!\!\!&.\!\!\!\!&.\!\!\!\!&.\!\!\!\!&.\!\!\!\!&.\!\!\!\!&2\!\!\!\!&1\!\!\!\!&1\!\!\!\!&2\!\!\!\!&.\!\!\!\!&.\!\!\!\!&.\!\!\!\!&.\!\!\!\!&1\!\!\!\!&2\!\!\!\!&.\!\!\!\!&2\\
.\!\!\!\!&.\!\!\!\!&1\!\!\!\!&.\!\!\!\!&.\!\!\!\!&.\!\!\!\!&.\!\!\!\!&.\!\!\!\!&.\!\!\!\!&.\!\!\!\!&.\!\!\!\!&.\!\!\!\!&.\!\!\!\!&.\!\!\!\!&.\!\!\!\!&.\!\!\!\!&.\!\!\!\!&.\!\!\!\!&.\!\!\!\!&.\!\!\!\!&.\!\!\!\!&.\!\!\!\!&.\!\!\!\!&.\!\!\!\!&.\!\!\!\!&.\!\!\!\!&.\!\!\!\!&.\!\!\!\!&.\!\!\!\!&.\!\!\!\!&2\!\!\!\!&.\!\!\!\!&2\!\!\!\!&.\!\!\!\!&.\!\!\!\!&.\!\!\!\!&1\!\!\!\!&1\!\!\!\!&1\!\!\!\!&2\!\!\!\!&.\!\!\!\!&1\!\!\!\!&.\!\!\!\!&1\!\!\!\!&1\!\!\!\!&1\!\!\!\!&1\!\!\!\!&1\!\!\!\!&1\!\!\!\!&1\!\!\!\!&.\!\!\!\!&2\!\!\!\!&2\!\!\!\!&1\!\!\!\!&1\!\!\!\!&.\!\!\!\!&.\!\!\!\!&1\!\!\!\!&.\!\!\!\!&1\\
.\!\!\!\!&.\!\!\!\!&.\!\!\!\!&1\!\!\!\!&.\!\!\!\!&.\!\!\!\!&.\!\!\!\!&.\!\!\!\!&.\!\!\!\!&.\!\!\!\!&.\!\!\!\!&.\!\!\!\!&.\!\!\!\!&.\!\!\!\!&.\!\!\!\!&.\!\!\!\!&.\!\!\!\!&.\!\!\!\!&.\!\!\!\!&.\!\!\!\!&.\!\!\!\!&.\!\!\!\!&.\!\!\!\!&.\!\!\!\!&.\!\!\!\!&.\!\!\!\!&.\!\!\!\!&.\!\!\!\!&.\!\!\!\!&.\!\!\!\!&.\!\!\!\!&.\!\!\!\!&.\!\!\!\!&2\!\!\!\!&1\!\!\!\!&2\!\!\!\!&.\!\!\!\!&2\!\!\!\!&.\!\!\!\!&2\!\!\!\!&.\!\!\!\!&2\!\!\!\!&2\!\!\!\!&1\!\!\!\!&.\!\!\!\!&2\!\!\!\!&.\!\!\!\!&1\!\!\!\!&.\!\!\!\!&2\!\!\!\!&.\!\!\!\!&2\!\!\!\!&.\!\!\!\!&.\!\!\!\!&2\!\!\!\!&2\!\!\!\!&1\!\!\!\!&2\!\!\!\!&.\!\!\!\!&2\\
.\!\!\!\!&.\!\!\!\!&.\!\!\!\!&.\!\!\!\!&1\!\!\!\!&.\!\!\!\!&.\!\!\!\!&.\!\!\!\!&.\!\!\!\!&.\!\!\!\!&.\!\!\!\!&.\!\!\!\!&.\!\!\!\!&.\!\!\!\!&.\!\!\!\!&.\!\!\!\!&.\!\!\!\!&.\!\!\!\!&.\!\!\!\!&.\!\!\!\!&.\!\!\!\!&.\!\!\!\!&.\!\!\!\!&.\!\!\!\!&.\!\!\!\!&.\!\!\!\!&.\!\!\!\!&.\!\!\!\!&.\!\!\!\!&.\!\!\!\!&.\!\!\!\!&1\!\!\!\!&2\!\!\!\!&1\!\!\!\!&1\!\!\!\!&1\!\!\!\!&.\!\!\!\!&1\!\!\!\!&.\!\!\!\!&2\!\!\!\!&.\!\!\!\!&1\!\!\!\!&2\!\!\!\!&.\!\!\!\!&1\!\!\!\!&1\!\!\!\!&.\!\!\!\!&1\!\!\!\!&2\!\!\!\!&1\!\!\!\!&.\!\!\!\!&1\!\!\!\!&.\!\!\!\!&2\!\!\!\!&2\!\!\!\!&.\!\!\!\!&1\!\!\!\!&.\!\!\!\!&2\!\!\!\!&2\\
.\!\!\!\!&.\!\!\!\!&.\!\!\!\!&.\!\!\!\!&.\!\!\!\!&1\!\!\!\!&.\!\!\!\!&.\!\!\!\!&.\!\!\!\!&.\!\!\!\!&.\!\!\!\!&.\!\!\!\!&.\!\!\!\!&.\!\!\!\!&.\!\!\!\!&.\!\!\!\!&.\!\!\!\!&.\!\!\!\!&.\!\!\!\!&.\!\!\!\!&.\!\!\!\!&.\!\!\!\!&.\!\!\!\!&.\!\!\!\!&.\!\!\!\!&.\!\!\!\!&.\!\!\!\!&.\!\!\!\!&.\!\!\!\!&.\!\!\!\!&1\!\!\!\!&1\!\!\!\!&.\!\!\!\!&.\!\!\!\!&2\!\!\!\!&.\!\!\!\!&1\!\!\!\!&1\!\!\!\!&1\!\!\!\!&2\!\!\!\!&.\!\!\!\!&2\!\!\!\!&1\!\!\!\!&2\!\!\!\!&1\!\!\!\!&2\!\!\!\!&2\!\!\!\!&2\!\!\!\!&.\!\!\!\!&1\!\!\!\!&2\!\!\!\!&1\!\!\!\!&2\!\!\!\!&1\!\!\!\!&.\!\!\!\!&.\!\!\!\!&.\!\!\!\!&2\!\!\!\!&.\!\!\!\!&.\\
.\!\!\!\!&.\!\!\!\!&.\!\!\!\!&.\!\!\!\!&.\!\!\!\!&.\!\!\!\!&1\!\!\!\!&.\!\!\!\!&.\!\!\!\!&.\!\!\!\!&.\!\!\!\!&.\!\!\!\!&.\!\!\!\!&.\!\!\!\!&.\!\!\!\!&.\!\!\!\!&.\!\!\!\!&.\!\!\!\!&.\!\!\!\!&.\!\!\!\!&.\!\!\!\!&.\!\!\!\!&.\!\!\!\!&.\!\!\!\!&.\!\!\!\!&.\!\!\!\!&.\!\!\!\!&.\!\!\!\!&.\!\!\!\!&.\!\!\!\!&.\!\!\!\!&.\!\!\!\!&.\!\!\!\!&2\!\!\!\!&1\!\!\!\!&1\!\!\!\!&.\!\!\!\!&.\!\!\!\!&2\!\!\!\!&2\!\!\!\!&.\!\!\!\!&.\!\!\!\!&.\!\!\!\!&2\!\!\!\!&1\!\!\!\!&1\!\!\!\!&2\!\!\!\!&1\!\!\!\!&1\!\!\!\!&.\!\!\!\!&1\!\!\!\!&2\!\!\!\!&1\!\!\!\!&.\!\!\!\!&1\!\!\!\!&2\!\!\!\!&1\!\!\!\!&.\!\!\!\!&.\!\!\!\!&.\\
.\!\!\!\!&.\!\!\!\!&.\!\!\!\!&.\!\!\!\!&.\!\!\!\!&.\!\!\!\!&.\!\!\!\!&1\!\!\!\!&.\!\!\!\!&.\!\!\!\!&.\!\!\!\!&.\!\!\!\!&.\!\!\!\!&.\!\!\!\!&.\!\!\!\!&.\!\!\!\!&.\!\!\!\!&.\!\!\!\!&.\!\!\!\!&.\!\!\!\!&.\!\!\!\!&.\!\!\!\!&.\!\!\!\!&.\!\!\!\!&.\!\!\!\!&.\!\!\!\!&.\!\!\!\!&.\!\!\!\!&.\!\!\!\!&.\!\!\!\!&1\!\!\!\!&1\!\!\!\!&.\!\!\!\!&1\!\!\!\!&2\!\!\!\!&.\!\!\!\!&1\!\!\!\!&2\!\!\!\!&1\!\!\!\!&2\!\!\!\!&2\!\!\!\!&2\!\!\!\!&1\!\!\!\!&.\!\!\!\!&1\!\!\!\!&2\!\!\!\!&1\!\!\!\!&.\!\!\!\!&2\!\!\!\!&1\!\!\!\!&2\!\!\!\!&1\!\!\!\!&.\!\!\!\!&.\!\!\!\!&2\!\!\!\!&1\!\!\!\!&2\!\!\!\!&2\!\!\!\!&.\!\!\!\!&2\\
.\!\!\!\!&.\!\!\!\!&.\!\!\!\!&.\!\!\!\!&.\!\!\!\!&.\!\!\!\!&.\!\!\!\!&.\!\!\!\!&1\!\!\!\!&.\!\!\!\!&.\!\!\!\!&.\!\!\!\!&.\!\!\!\!&.\!\!\!\!&.\!\!\!\!&.\!\!\!\!&.\!\!\!\!&.\!\!\!\!&.\!\!\!\!&.\!\!\!\!&.\!\!\!\!&.\!\!\!\!&.\!\!\!\!&.\!\!\!\!&.\!\!\!\!&.\!\!\!\!&.\!\!\!\!&.\!\!\!\!&.\!\!\!\!&.\!\!\!\!&1\!\!\!\!&.\!\!\!\!&.\!\!\!\!&1\!\!\!\!&1\!\!\!\!&2\!\!\!\!&2\!\!\!\!&.\!\!\!\!&1\!\!\!\!&2\!\!\!\!&1\!\!\!\!&.\!\!\!\!&1\!\!\!\!&2\!\!\!\!&2\!\!\!\!&.\!\!\!\!&1\!\!\!\!&.\!\!\!\!&.\!\!\!\!&1\!\!\!\!&.\!\!\!\!&1\!\!\!\!&.\!\!\!\!&.\!\!\!\!&.\!\!\!\!&2\!\!\!\!&2\!\!\!\!&.\!\!\!\!&.\!\!\!\!&1\\
.\!\!\!\!&.\!\!\!\!&.\!\!\!\!&.\!\!\!\!&.\!\!\!\!&.\!\!\!\!&.\!\!\!\!&.\!\!\!\!&.\!\!\!\!&1\!\!\!\!&.\!\!\!\!&.\!\!\!\!&.\!\!\!\!&.\!\!\!\!&.\!\!\!\!&.\!\!\!\!&.\!\!\!\!&.\!\!\!\!&.\!\!\!\!&.\!\!\!\!&.\!\!\!\!&.\!\!\!\!&.\!\!\!\!&.\!\!\!\!&.\!\!\!\!&.\!\!\!\!&.\!\!\!\!&.\!\!\!\!&.\!\!\!\!&.\!\!\!\!&1\!\!\!\!&2\!\!\!\!&1\!\!\!\!&2\!\!\!\!&1\!\!\!\!&2\!\!\!\!&.\!\!\!\!&1\!\!\!\!&.\!\!\!\!&2\!\!\!\!&2\!\!\!\!&2\!\!\!\!&2\!\!\!\!&.\!\!\!\!&1\!\!\!\!&2\!\!\!\!&2\!\!\!\!&2\!\!\!\!&1\!\!\!\!&1\!\!\!\!&1\!\!\!\!&2\!\!\!\!&.\!\!\!\!&.\!\!\!\!&.\!\!\!\!&.\!\!\!\!&1\!\!\!\!&1\!\!\!\!&2\!\!\!\!&2\\
.\!\!\!\!&.\!\!\!\!&.\!\!\!\!&.\!\!\!\!&.\!\!\!\!&.\!\!\!\!&.\!\!\!\!&.\!\!\!\!&.\!\!\!\!&.\!\!\!\!&1\!\!\!\!&.\!\!\!\!&.\!\!\!\!&.\!\!\!\!&.\!\!\!\!&.\!\!\!\!&.\!\!\!\!&.\!\!\!\!&.\!\!\!\!&.\!\!\!\!&.\!\!\!\!&.\!\!\!\!&.\!\!\!\!&.\!\!\!\!&.\!\!\!\!&.\!\!\!\!&.\!\!\!\!&.\!\!\!\!&.\!\!\!\!&.\!\!\!\!&.\!\!\!\!&1\!\!\!\!&1\!\!\!\!&.\!\!\!\!&.\!\!\!\!&.\!\!\!\!&.\!\!\!\!&1\!\!\!\!&1\!\!\!\!&.\!\!\!\!&.\!\!\!\!&1\!\!\!\!&2\!\!\!\!&1\!\!\!\!&.\!\!\!\!&.\!\!\!\!&1\!\!\!\!&2\!\!\!\!&1\!\!\!\!&.\!\!\!\!&1\!\!\!\!&.\!\!\!\!&.\!\!\!\!&2\!\!\!\!&1\!\!\!\!&.\!\!\!\!&1\!\!\!\!&1\!\!\!\!&1\!\!\!\!&1\\
.\!\!\!\!&.\!\!\!\!&.\!\!\!\!&.\!\!\!\!&.\!\!\!\!&.\!\!\!\!&.\!\!\!\!&.\!\!\!\!&.\!\!\!\!&.\!\!\!\!&.\!\!\!\!&1\!\!\!\!&.\!\!\!\!&.\!\!\!\!&.\!\!\!\!&.\!\!\!\!&.\!\!\!\!&.\!\!\!\!&.\!\!\!\!&.\!\!\!\!&.\!\!\!\!&.\!\!\!\!&.\!\!\!\!&.\!\!\!\!&.\!\!\!\!&.\!\!\!\!&.\!\!\!\!&.\!\!\!\!&.\!\!\!\!&.\!\!\!\!&2\!\!\!\!&2\!\!\!\!&.\!\!\!\!&1\!\!\!\!&1\!\!\!\!&2\!\!\!\!&2\!\!\!\!&1\!\!\!\!&.\!\!\!\!&.\!\!\!\!&2\!\!\!\!&1\!\!\!\!&.\!\!\!\!&1\!\!\!\!&.\!\!\!\!&.\!\!\!\!&2\!\!\!\!&.\!\!\!\!&1\!\!\!\!&1\!\!\!\!&1\!\!\!\!&1\!\!\!\!&2\!\!\!\!&1\!\!\!\!&1\!\!\!\!&.\!\!\!\!&.\!\!\!\!&2\!\!\!\!&.\!\!\!\!&1\\
.\!\!\!\!&.\!\!\!\!&.\!\!\!\!&.\!\!\!\!&.\!\!\!\!&.\!\!\!\!&.\!\!\!\!&.\!\!\!\!&.\!\!\!\!&.\!\!\!\!&.\!\!\!\!&.\!\!\!\!&1\!\!\!\!&.\!\!\!\!&.\!\!\!\!&.\!\!\!\!&.\!\!\!\!&.\!\!\!\!&.\!\!\!\!&.\!\!\!\!&.\!\!\!\!&.\!\!\!\!&.\!\!\!\!&.\!\!\!\!&.\!\!\!\!&.\!\!\!\!&.\!\!\!\!&.\!\!\!\!&.\!\!\!\!&.\!\!\!\!&2\!\!\!\!&1\!\!\!\!&.\!\!\!\!&2\!\!\!\!&2\!\!\!\!&2\!\!\!\!&.\!\!\!\!&2\!\!\!\!&2\!\!\!\!&.\!\!\!\!&2\!\!\!\!&1\!\!\!\!&1\!\!\!\!&1\!\!\!\!&.\!\!\!\!&2\!\!\!\!&.\!\!\!\!&1\!\!\!\!&2\!\!\!\!&.\!\!\!\!&.\!\!\!\!&2\!\!\!\!&1\!\!\!\!&.\!\!\!\!&1\!\!\!\!&1\!\!\!\!&.\!\!\!\!&2\!\!\!\!&2\!\!\!\!&.\\
.\!\!\!\!&.\!\!\!\!&.\!\!\!\!&.\!\!\!\!&.\!\!\!\!&.\!\!\!\!&.\!\!\!\!&.\!\!\!\!&.\!\!\!\!&.\!\!\!\!&.\!\!\!\!&.\!\!\!\!&.\!\!\!\!&1\!\!\!\!&.\!\!\!\!&.\!\!\!\!&.\!\!\!\!&.\!\!\!\!&.\!\!\!\!&.\!\!\!\!&.\!\!\!\!&.\!\!\!\!&.\!\!\!\!&.\!\!\!\!&.\!\!\!\!&.\!\!\!\!&.\!\!\!\!&.\!\!\!\!&.\!\!\!\!&.\!\!\!\!&1\!\!\!\!&.\!\!\!\!&.\!\!\!\!&2\!\!\!\!&1\!\!\!\!&.\!\!\!\!&1\!\!\!\!&2\!\!\!\!&.\!\!\!\!&.\!\!\!\!&2\!\!\!\!&.\!\!\!\!&2\!\!\!\!&.\!\!\!\!&1\!\!\!\!&.\!\!\!\!&1\!\!\!\!&2\!\!\!\!&.\!\!\!\!&2\!\!\!\!&1\!\!\!\!&1\!\!\!\!&.\!\!\!\!&1\!\!\!\!&1\!\!\!\!&.\!\!\!\!&1\!\!\!\!&.\!\!\!\!&2\!\!\!\!&.\\
.\!\!\!\!&.\!\!\!\!&.\!\!\!\!&.\!\!\!\!&.\!\!\!\!&.\!\!\!\!&.\!\!\!\!&.\!\!\!\!&.\!\!\!\!&.\!\!\!\!&.\!\!\!\!&.\!\!\!\!&.\!\!\!\!&.\!\!\!\!&1\!\!\!\!&.\!\!\!\!&.\!\!\!\!&.\!\!\!\!&.\!\!\!\!&.\!\!\!\!&.\!\!\!\!&.\!\!\!\!&.\!\!\!\!&.\!\!\!\!&.\!\!\!\!&.\!\!\!\!&.\!\!\!\!&.\!\!\!\!&.\!\!\!\!&.\!\!\!\!&2\!\!\!\!&2\!\!\!\!&1\!\!\!\!&2\!\!\!\!&.\!\!\!\!&1\!\!\!\!&1\!\!\!\!&.\!\!\!\!&1\!\!\!\!&.\!\!\!\!&1\!\!\!\!&1\!\!\!\!&1\!\!\!\!&1\!\!\!\!&1\!\!\!\!&1\!\!\!\!&.\!\!\!\!&1\!\!\!\!&2\!\!\!\!&.\!\!\!\!&.\!\!\!\!&2\!\!\!\!&.\!\!\!\!&.\!\!\!\!&1\!\!\!\!&2\!\!\!\!&1\!\!\!\!&.\!\!\!\!&.\!\!\!\!&1\\
.\!\!\!\!&.\!\!\!\!&.\!\!\!\!&.\!\!\!\!&.\!\!\!\!&.\!\!\!\!&.\!\!\!\!&.\!\!\!\!&.\!\!\!\!&.\!\!\!\!&.\!\!\!\!&.\!\!\!\!&.\!\!\!\!&.\!\!\!\!&.\!\!\!\!&1\!\!\!\!&.\!\!\!\!&.\!\!\!\!&.\!\!\!\!&.\!\!\!\!&.\!\!\!\!&.\!\!\!\!&.\!\!\!\!&.\!\!\!\!&.\!\!\!\!&.\!\!\!\!&.\!\!\!\!&.\!\!\!\!&.\!\!\!\!&.\!\!\!\!&2\!\!\!\!&1\!\!\!\!&.\!\!\!\!&2\!\!\!\!&.\!\!\!\!&2\!\!\!\!&1\!\!\!\!&.\!\!\!\!&2\!\!\!\!&.\!\!\!\!&1\!\!\!\!&1\!\!\!\!&.\!\!\!\!&.\!\!\!\!&.\!\!\!\!&1\!\!\!\!&.\!\!\!\!&.\!\!\!\!&.\!\!\!\!&1\!\!\!\!&1\!\!\!\!&2\!\!\!\!&1\!\!\!\!&.\!\!\!\!&2\!\!\!\!&.\!\!\!\!&2\!\!\!\!&2\!\!\!\!&.\!\!\!\!&1\\
.\!\!\!\!&.\!\!\!\!&.\!\!\!\!&.\!\!\!\!&.\!\!\!\!&.\!\!\!\!&.\!\!\!\!&.\!\!\!\!&.\!\!\!\!&.\!\!\!\!&.\!\!\!\!&.\!\!\!\!&.\!\!\!\!&.\!\!\!\!&.\!\!\!\!&.\!\!\!\!&1\!\!\!\!&.\!\!\!\!&.\!\!\!\!&.\!\!\!\!&.\!\!\!\!&.\!\!\!\!&.\!\!\!\!&.\!\!\!\!&.\!\!\!\!&.\!\!\!\!&.\!\!\!\!&.\!\!\!\!&.\!\!\!\!&.\!\!\!\!&.\!\!\!\!&2\!\!\!\!&1\!\!\!\!&1\!\!\!\!&1\!\!\!\!&.\!\!\!\!&1\!\!\!\!&2\!\!\!\!&2\!\!\!\!&1\!\!\!\!&2\!\!\!\!&2\!\!\!\!&2\!\!\!\!&2\!\!\!\!&.\!\!\!\!&2\!\!\!\!&.\!\!\!\!&.\!\!\!\!&2\!\!\!\!&1\!\!\!\!&.\!\!\!\!&.\!\!\!\!&1\!\!\!\!&2\!\!\!\!&2\!\!\!\!&2\!\!\!\!&.\!\!\!\!&.\!\!\!\!&1\!\!\!\!&.\\
.\!\!\!\!&.\!\!\!\!&.\!\!\!\!&.\!\!\!\!&.\!\!\!\!&.\!\!\!\!&.\!\!\!\!&.\!\!\!\!&.\!\!\!\!&.\!\!\!\!&.\!\!\!\!&.\!\!\!\!&.\!\!\!\!&.\!\!\!\!&.\!\!\!\!&.\!\!\!\!&.\!\!\!\!&1\!\!\!\!&.\!\!\!\!&.\!\!\!\!&.\!\!\!\!&.\!\!\!\!&.\!\!\!\!&.\!\!\!\!&.\!\!\!\!&.\!\!\!\!&.\!\!\!\!&.\!\!\!\!&.\!\!\!\!&.\!\!\!\!&.\!\!\!\!&1\!\!\!\!&2\!\!\!\!&.\!\!\!\!&.\!\!\!\!&1\!\!\!\!&1\!\!\!\!&.\!\!\!\!&1\!\!\!\!&2\!\!\!\!&.\!\!\!\!&2\!\!\!\!&.\!\!\!\!&.\!\!\!\!&2\!\!\!\!&.\!\!\!\!&.\!\!\!\!&.\!\!\!\!&.\!\!\!\!&1\!\!\!\!&2\!\!\!\!&1\!\!\!\!&1\!\!\!\!&2\!\!\!\!&.\!\!\!\!&.\!\!\!\!&1\!\!\!\!&2\!\!\!\!&2\!\!\!\!&1\\
.\!\!\!\!&.\!\!\!\!&.\!\!\!\!&.\!\!\!\!&.\!\!\!\!&.\!\!\!\!&.\!\!\!\!&.\!\!\!\!&.\!\!\!\!&.\!\!\!\!&.\!\!\!\!&.\!\!\!\!&.\!\!\!\!&.\!\!\!\!&.\!\!\!\!&.\!\!\!\!&.\!\!\!\!&.\!\!\!\!&1\!\!\!\!&.\!\!\!\!&.\!\!\!\!&.\!\!\!\!&.\!\!\!\!&.\!\!\!\!&.\!\!\!\!&.\!\!\!\!&.\!\!\!\!&.\!\!\!\!&.\!\!\!\!&.\!\!\!\!&2\!\!\!\!&.\!\!\!\!&.\!\!\!\!&.\!\!\!\!&2\!\!\!\!&.\!\!\!\!&2\!\!\!\!&1\!\!\!\!&1\!\!\!\!&1\!\!\!\!&2\!\!\!\!&1\!\!\!\!&.\!\!\!\!&1\!\!\!\!&2\!\!\!\!&.\!\!\!\!&2\!\!\!\!&1\!\!\!\!&2\!\!\!\!&.\!\!\!\!&2\!\!\!\!&1\!\!\!\!&2\!\!\!\!&1\!\!\!\!&.\!\!\!\!&.\!\!\!\!&1\!\!\!\!&1\!\!\!\!&2\!\!\!\!&.\\
.\!\!\!\!&.\!\!\!\!&.\!\!\!\!&.\!\!\!\!&.\!\!\!\!&.\!\!\!\!&.\!\!\!\!&.\!\!\!\!&.\!\!\!\!&.\!\!\!\!&.\!\!\!\!&.\!\!\!\!&.\!\!\!\!&.\!\!\!\!&.\!\!\!\!&.\!\!\!\!&.\!\!\!\!&.\!\!\!\!&.\!\!\!\!&1\!\!\!\!&.\!\!\!\!&.\!\!\!\!&.\!\!\!\!&.\!\!\!\!&.\!\!\!\!&.\!\!\!\!&.\!\!\!\!&.\!\!\!\!&.\!\!\!\!&.\!\!\!\!&2\!\!\!\!&.\!\!\!\!&2\!\!\!\!&2\!\!\!\!&.\!\!\!\!&2\!\!\!\!&.\!\!\!\!&2\!\!\!\!&1\!\!\!\!&2\!\!\!\!&2\!\!\!\!&2\!\!\!\!&2\!\!\!\!&2\!\!\!\!&2\!\!\!\!&2\!\!\!\!&1\!\!\!\!&.\!\!\!\!&2\!\!\!\!&1\!\!\!\!&2\!\!\!\!&2\!\!\!\!&1\!\!\!\!&2\!\!\!\!&1\!\!\!\!&1\!\!\!\!&2\!\!\!\!&2\!\!\!\!&2\!\!\!\!&1\\
.\!\!\!\!&.\!\!\!\!&.\!\!\!\!&.\!\!\!\!&.\!\!\!\!&.\!\!\!\!&.\!\!\!\!&.\!\!\!\!&.\!\!\!\!&.\!\!\!\!&.\!\!\!\!&.\!\!\!\!&.\!\!\!\!&.\!\!\!\!&.\!\!\!\!&.\!\!\!\!&.\!\!\!\!&.\!\!\!\!&.\!\!\!\!&.\!\!\!\!&1\!\!\!\!&.\!\!\!\!&.\!\!\!\!&.\!\!\!\!&.\!\!\!\!&.\!\!\!\!&.\!\!\!\!&.\!\!\!\!&.\!\!\!\!&.\!\!\!\!&.\!\!\!\!&2\!\!\!\!&.\!\!\!\!&.\!\!\!\!&.\!\!\!\!&.\!\!\!\!&1\!\!\!\!&1\!\!\!\!&1\!\!\!\!&.\!\!\!\!&1\!\!\!\!&.\!\!\!\!&2\!\!\!\!&1\!\!\!\!&1\!\!\!\!&2\!\!\!\!&1\!\!\!\!&1\!\!\!\!&1\!\!\!\!&1\!\!\!\!&1\!\!\!\!&2\!\!\!\!&2\!\!\!\!&1\!\!\!\!&2\!\!\!\!&1\!\!\!\!&2\!\!\!\!&1\!\!\!\!&2\!\!\!\!&1\\
.\!\!\!\!&.\!\!\!\!&.\!\!\!\!&.\!\!\!\!&.\!\!\!\!&.\!\!\!\!&.\!\!\!\!&.\!\!\!\!&.\!\!\!\!&.\!\!\!\!&.\!\!\!\!&.\!\!\!\!&.\!\!\!\!&.\!\!\!\!&.\!\!\!\!&.\!\!\!\!&.\!\!\!\!&.\!\!\!\!&.\!\!\!\!&.\!\!\!\!&.\!\!\!\!&1\!\!\!\!&.\!\!\!\!&.\!\!\!\!&.\!\!\!\!&.\!\!\!\!&.\!\!\!\!&.\!\!\!\!&.\!\!\!\!&.\!\!\!\!&.\!\!\!\!&2\!\!\!\!&.\!\!\!\!&2\!\!\!\!&2\!\!\!\!&2\!\!\!\!&.\!\!\!\!&1\!\!\!\!&2\!\!\!\!&1\!\!\!\!&.\!\!\!\!&.\!\!\!\!&2\!\!\!\!&1\!\!\!\!&1\!\!\!\!&.\!\!\!\!&2\!\!\!\!&2\!\!\!\!&2\!\!\!\!&1\!\!\!\!&1\!\!\!\!&2\!\!\!\!&2\!\!\!\!&.\!\!\!\!&2\!\!\!\!&.\!\!\!\!&2\!\!\!\!&.\!\!\!\!&1\!\!\!\!&.\\
.\!\!\!\!&.\!\!\!\!&.\!\!\!\!&.\!\!\!\!&.\!\!\!\!&.\!\!\!\!&.\!\!\!\!&.\!\!\!\!&.\!\!\!\!&.\!\!\!\!&.\!\!\!\!&.\!\!\!\!&.\!\!\!\!&.\!\!\!\!&.\!\!\!\!&.\!\!\!\!&.\!\!\!\!&.\!\!\!\!&.\!\!\!\!&.\!\!\!\!&.\!\!\!\!&.\!\!\!\!&1\!\!\!\!&.\!\!\!\!&.\!\!\!\!&.\!\!\!\!&.\!\!\!\!&.\!\!\!\!&.\!\!\!\!&.\!\!\!\!&2\!\!\!\!&2\!\!\!\!&1\!\!\!\!&2\!\!\!\!&1\!\!\!\!&2\!\!\!\!&2\!\!\!\!&2\!\!\!\!&1\!\!\!\!&2\!\!\!\!&2\!\!\!\!&1\!\!\!\!&1\!\!\!\!&2\!\!\!\!&1\!\!\!\!&2\!\!\!\!&.\!\!\!\!&.\!\!\!\!&.\!\!\!\!&2\!\!\!\!&.\!\!\!\!&1\!\!\!\!&.\!\!\!\!&.\!\!\!\!&.\!\!\!\!&.\!\!\!\!&2\!\!\!\!&.\!\!\!\!&.\!\!\!\!&2\\
.\!\!\!\!&.\!\!\!\!&.\!\!\!\!&.\!\!\!\!&.\!\!\!\!&.\!\!\!\!&.\!\!\!\!&.\!\!\!\!&.\!\!\!\!&.\!\!\!\!&.\!\!\!\!&.\!\!\!\!&.\!\!\!\!&.\!\!\!\!&.\!\!\!\!&.\!\!\!\!&.\!\!\!\!&.\!\!\!\!&.\!\!\!\!&.\!\!\!\!&.\!\!\!\!&.\!\!\!\!&.\!\!\!\!&1\!\!\!\!&.\!\!\!\!&.\!\!\!\!&.\!\!\!\!&.\!\!\!\!&.\!\!\!\!&.\!\!\!\!&2\!\!\!\!&1\!\!\!\!&.\!\!\!\!&.\!\!\!\!&.\!\!\!\!&2\!\!\!\!&2\!\!\!\!&2\!\!\!\!&1\!\!\!\!&.\!\!\!\!&1\!\!\!\!&1\!\!\!\!&.\!\!\!\!&.\!\!\!\!&2\!\!\!\!&2\!\!\!\!&2\!\!\!\!&1\!\!\!\!&.\!\!\!\!&1\!\!\!\!&2\!\!\!\!&2\!\!\!\!&.\!\!\!\!&1\!\!\!\!&1\!\!\!\!&1\!\!\!\!&2\!\!\!\!&.\!\!\!\!&2\!\!\!\!&.\\
.\!\!\!\!&.\!\!\!\!&.\!\!\!\!&.\!\!\!\!&.\!\!\!\!&.\!\!\!\!&.\!\!\!\!&.\!\!\!\!&.\!\!\!\!&.\!\!\!\!&.\!\!\!\!&.\!\!\!\!&.\!\!\!\!&.\!\!\!\!&.\!\!\!\!&.\!\!\!\!&.\!\!\!\!&.\!\!\!\!&.\!\!\!\!&.\!\!\!\!&.\!\!\!\!&.\!\!\!\!&.\!\!\!\!&.\!\!\!\!&1\!\!\!\!&.\!\!\!\!&.\!\!\!\!&.\!\!\!\!&.\!\!\!\!&.\!\!\!\!&1\!\!\!\!&1\!\!\!\!&2\!\!\!\!&1\!\!\!\!&1\!\!\!\!&2\!\!\!\!&1\!\!\!\!&2\!\!\!\!&.\!\!\!\!&.\!\!\!\!&2\!\!\!\!&.\!\!\!\!&.\!\!\!\!&2\!\!\!\!&.\!\!\!\!&.\!\!\!\!&2\!\!\!\!&.\!\!\!\!&1\!\!\!\!&2\!\!\!\!&1\!\!\!\!&.\!\!\!\!&1\!\!\!\!&.\!\!\!\!&1\!\!\!\!&1\!\!\!\!&2\!\!\!\!&.\!\!\!\!&2\!\!\!\!&2\\
.\!\!\!\!&.\!\!\!\!&.\!\!\!\!&.\!\!\!\!&.\!\!\!\!&.\!\!\!\!&.\!\!\!\!&.\!\!\!\!&.\!\!\!\!&.\!\!\!\!&.\!\!\!\!&.\!\!\!\!&.\!\!\!\!&.\!\!\!\!&.\!\!\!\!&.\!\!\!\!&.\!\!\!\!&.\!\!\!\!&.\!\!\!\!&.\!\!\!\!&.\!\!\!\!&.\!\!\!\!&.\!\!\!\!&.\!\!\!\!&.\!\!\!\!&1\!\!\!\!&.\!\!\!\!&.\!\!\!\!&.\!\!\!\!&.\!\!\!\!&.\!\!\!\!&2\!\!\!\!&2\!\!\!\!&2\!\!\!\!&1\!\!\!\!&2\!\!\!\!&1\!\!\!\!&1\!\!\!\!&2\!\!\!\!&.\!\!\!\!&.\!\!\!\!&2\!\!\!\!&.\!\!\!\!&1\!\!\!\!&1\!\!\!\!&.\!\!\!\!&1\!\!\!\!&.\!\!\!\!&2\!\!\!\!&.\!\!\!\!&2\!\!\!\!&2\!\!\!\!&.\!\!\!\!&1\!\!\!\!&1\!\!\!\!&2\!\!\!\!&1\!\!\!\!&.\!\!\!\!&2\!\!\!\!&.\\
.\!\!\!\!&.\!\!\!\!&.\!\!\!\!&.\!\!\!\!&.\!\!\!\!&.\!\!\!\!&.\!\!\!\!&.\!\!\!\!&.\!\!\!\!&.\!\!\!\!&.\!\!\!\!&.\!\!\!\!&.\!\!\!\!&.\!\!\!\!&.\!\!\!\!&.\!\!\!\!&.\!\!\!\!&.\!\!\!\!&.\!\!\!\!&.\!\!\!\!&.\!\!\!\!&.\!\!\!\!&.\!\!\!\!&.\!\!\!\!&.\!\!\!\!&.\!\!\!\!&1\!\!\!\!&.\!\!\!\!&.\!\!\!\!&.\!\!\!\!&.\!\!\!\!&.\!\!\!\!&2\!\!\!\!&2\!\!\!\!&1\!\!\!\!&1\!\!\!\!&.\!\!\!\!&2\!\!\!\!&1\!\!\!\!&.\!\!\!\!&.\!\!\!\!&2\!\!\!\!&2\!\!\!\!&.\!\!\!\!&2\!\!\!\!&2\!\!\!\!&.\!\!\!\!&2\!\!\!\!&2\!\!\!\!&2\!\!\!\!&1\!\!\!\!&1\!\!\!\!&2\!\!\!\!&.\!\!\!\!&2\!\!\!\!&1\!\!\!\!&2\!\!\!\!&.\!\!\!\!&2\!\!\!\!&.\\
.\!\!\!\!&.\!\!\!\!&.\!\!\!\!&.\!\!\!\!&.\!\!\!\!&.\!\!\!\!&.\!\!\!\!&.\!\!\!\!&.\!\!\!\!&.\!\!\!\!&.\!\!\!\!&.\!\!\!\!&.\!\!\!\!&.\!\!\!\!&.\!\!\!\!&.\!\!\!\!&.\!\!\!\!&.\!\!\!\!&.\!\!\!\!&.\!\!\!\!&.\!\!\!\!&.\!\!\!\!&.\!\!\!\!&.\!\!\!\!&.\!\!\!\!&.\!\!\!\!&.\!\!\!\!&1\!\!\!\!&.\!\!\!\!&.\!\!\!\!&2\!\!\!\!&.\!\!\!\!&2\!\!\!\!&.\!\!\!\!&.\!\!\!\!&.\!\!\!\!&1\!\!\!\!&1\!\!\!\!&2\!\!\!\!&1\!\!\!\!&2\!\!\!\!&1\!\!\!\!&2\!\!\!\!&.\!\!\!\!&.\!\!\!\!&2\!\!\!\!&2\!\!\!\!&.\!\!\!\!&2\!\!\!\!&2\!\!\!\!&2\!\!\!\!&1\!\!\!\!&1\!\!\!\!&.\!\!\!\!&1\!\!\!\!&2\!\!\!\!&.\!\!\!\!&2\!\!\!\!&2\!\!\!\!&.\\
.\!\!\!\!&.\!\!\!\!&.\!\!\!\!&.\!\!\!\!&.\!\!\!\!&.\!\!\!\!&.\!\!\!\!&.\!\!\!\!&.\!\!\!\!&.\!\!\!\!&.\!\!\!\!&.\!\!\!\!&.\!\!\!\!&.\!\!\!\!&.\!\!\!\!&.\!\!\!\!&.\!\!\!\!&.\!\!\!\!&.\!\!\!\!&.\!\!\!\!&.\!\!\!\!&.\!\!\!\!&.\!\!\!\!&.\!\!\!\!&.\!\!\!\!&.\!\!\!\!&.\!\!\!\!&.\!\!\!\!&1\!\!\!\!&.\!\!\!\!&1\!\!\!\!&.\!\!\!\!&1\!\!\!\!&2\!\!\!\!&1\!\!\!\!&1\!\!\!\!&.\!\!\!\!&2\!\!\!\!&1\!\!\!\!&2\!\!\!\!&2\!\!\!\!&.\!\!\!\!&.\!\!\!\!&1\!\!\!\!&2\!\!\!\!&1\!\!\!\!&2\!\!\!\!&1\!\!\!\!&.\!\!\!\!&.\!\!\!\!&.\!\!\!\!&1\!\!\!\!&2\!\!\!\!&1\!\!\!\!&2\!\!\!\!&2\!\!\!\!&1\!\!\!\!&1\!\!\!\!&1\!\!\!\!&2\\
.\!\!\!\!&.\!\!\!\!&.\!\!\!\!&.\!\!\!\!&.\!\!\!\!&.\!\!\!\!&.\!\!\!\!&.\!\!\!\!&.\!\!\!\!&.\!\!\!\!&.\!\!\!\!&.\!\!\!\!&.\!\!\!\!&.\!\!\!\!&.\!\!\!\!&.\!\!\!\!&.\!\!\!\!&.\!\!\!\!&.\!\!\!\!&.\!\!\!\!&.\!\!\!\!&.\!\!\!\!&.\!\!\!\!&.\!\!\!\!&.\!\!\!\!&.\!\!\!\!&.\!\!\!\!&.\!\!\!\!&.\!\!\!\!&1\!\!\!\!&1\!\!\!\!&1\!\!\!\!&1\!\!\!\!&2\!\!\!\!&1\!\!\!\!&2\!\!\!\!&1\!\!\!\!&.\!\!\!\!&1\!\!\!\!&.\!\!\!\!&1\!\!\!\!&1\!\!\!\!&.\!\!\!\!&1\!\!\!\!&.\!\!\!\!&1\!\!\!\!&1\!\!\!\!&.\!\!\!\!&.\!\!\!\!&1\!\!\!\!&.\!\!\!\!&.\!\!\!\!&.\!\!\!\!&.\!\!\!\!&1\!\!\!\!&2\!\!\!\!&2\!\!\!\!&.\!\!\!\!&.\!\!\!\!&.
  \end{array} \right)\hspace{-1in} $$
  \normalsize
  Since $5$ is coprime to $24$, we do not need any automorphism. \hfill\ensuremath{\Diamond}
\end{enumerate}
\end{examplenodiamond}


\newcommand{\etalchar}[1]{$^{#1}$}

\end{document}